%% file: main.tex
\pdfoutput=1
\documentclass[11pt,letterpaper]{article}
\usepackage{fullpage}
\usepackage[round]{natbib}
\usepackage{amsmath}
\usepackage{amsfonts}
\usepackage{hyperref}
\usepackage{booktabs} % For formal tables
\usepackage[ruled,noend]{algorithm2e} % For algorithms
\renewcommand{\algorithmcfname}{ALGORITHM}
\SetAlFnt{\small}
\SetAlCapFnt{\small}
\SetAlCapNameFnt{\small}
\SetAlCapHSkip{0pt}
\IncMargin{-\parindent}

\numberwithin{equation}{section}
\usepackage{slivkins-theorems}

\newenvironment{OneLiners}[1][\ensuremath{\bullet}]
    {\begin{list}
        {#1}
        {\setlength{\itemsep}{0pt}
        \setlength{\parsep}{0pt}
          \setlength{\topsep }{0pt}
        }}
    {\end{list}}

\input{mymath}

\usepackage{mathtools}
\usepackage[thinc]{esdiff}
\DeclarePairedDelimiter\ceil{\lceil}{\rceil}
\usepackage{bbm}
\usepackage{xcolor}
\usepackage{paralist}

\newcommand\defeq{\mathrel{:=}}
\newcommand{\muhat}{{\hat{\mu}}}

\newcommand{\term}[1]{\ensuremath{\mathtt{#1}}\xspace}
\newcommand{\BSL}{\term{BSL}}

\newcommand{\hist}{\term{hist}} % history
\newcommand{\pfail}{p_{\term{fail}}} % failure prob
\newcommand{\qfail}{q_{\term{fail}}} % failure prob
\newcommand{\sampfail}{\term{Fail}}  % failure event
\newcommand{\clean}{\term{Clean}}
\newcommand{\indx}{\term{Ind}}

\newcommand{\ucb}[1][\eta]{\term{UCB}^{#1}}
\newcommand{\lcb}[1][\eta]{\term{LCB}^{#1}}
\newcommand{\regret}{\textnormal{Regret}}

\newcommand{\etamax}{\eta_{\max}} % for the UB section

\newcommand{\tape}{\term{Tape}}
\newcommand{\ucbTape}[1][\eta]{\term{UCB}^{\term{tape},\,#1}}
\newcommand{\lcbTape}[1][\eta]{\term{LCB}^{\term{tape},\,#1}}
\newcommand{\muTape}{\widehat{\mu}^{\term{tape}}}

\newcommand{\setnot}[1]{\left\{ #1 \right\}}

\renewcommand{\paragraph}[1]{\smallskip \noindent{\bf #1.} }
\newcommand{\EXP}{\mathbb{E}}
\newcommand{\PROB}{\mathrm{Pr}}

\DeclareMathOperator*{\argmax}{argmax}

\renewcommand{\Pr}[1]{{\PROB \sbr{#1}}}

\newcommand{\Ex}[1]{{\EXP \sbr{#1}}}

\newcommand{\eps}{\varepsilon}
\newcommand{\F}{\mathcal{F}}
\newcommand{\bernoul}{{\text{Bernoulli}}}

\usepackage{color-edits}
\addauthor{as}{red}      % as for Alex
\addauthor{kb}{cyan}   % kb for Kiarash
\addauthor{ss}{magenta} % ss for Suho
\addauthor{mh}{magenta} % mh for Mohammad
    \newcommand{\ie}{{\em i.e.,~\xspace}}
    
    \newcommand{\eg}{{\em e.g.,~\xspace}}
    \newcommand{\Eg}{{\em E.g.,~\xspace}}

    \newcommand{\rbr}[1]{\left(\,#1\,\right)}
    \newcommand{\abs}[1]{\left|\,#1\,\right|}
    \newcommand{\sbr}[1]{\left[\,#1\,\right]}
    \newcommand{\cbr}[1]{\left\{\,#1\,\right\}}

    \newcommand{\xhdr}[1]{\vspace{2mm} \noindent{\bf #1}}

    \usepackage{nicefrac}
    \newcommand{\EqComment}[1]{\text{\emph{(#1)}}}
    \newcommand{\LDOTS}{\, ,\ \ldots\ ,}     % smart "..."
    \newcommand{\geqfosd}{\succeq_{\term{fosd}}}

    \newcommand{\dd}{\;\mathtt{d}}
    \renewcommand{\eqref}[1]{(\ref{#1})}
    \renewcommand{\refeq}[1]{Eq.~(\ref{#1})}

    \newcommand{\pgap}{\Delta_{\cP}} % prior gap
    \newcommand{\GoodSmall}{\cE_{1}^{\mathtt{UB}}} % "good arm has small mu1"
    \newcommand{\BadLarge}[1][2]{\cE_{#1}^{\mathtt{LB}}}  % "bad arm has large mu2"
    \newcommand{\minPDF}{\cP_{\min}} % LB on pdf of the prior

\title{Bandit Social Learning: Exploration under Myopic Behavior%
\footnote{
A preliminary version of this paper has been published in NeurIPS 2023, titled ``Bandit Social Learning under Myopic Behavior". Since Nov'23, this paper features several new results compared to the NeurIPS version. Specifically, we added Section~\ref{sec:K} (on $K\geq 2$ arms) and Section~\ref{sec:expts} (simulations), generalized Corollary~\ref{cor:BG} from independent to correlated priors, and strengthened the main ``negative" guarantees (in Section~\ref{sec:fail}).

\vspace{2mm}\indent
Early versions of our results on the greedy algorithm (Corollary~\ref{cor:LB-unbiased} and Theorem~\ref{thm:BG}) have been available in a book chapter by A. Slivkins \citep[Ch. 11]{slivkins-MABbook}. The authors acknowledge Mark Sellke for proving Theorem~\ref{thm:BG} and suggesting a proof plan for a version of Corollary~\ref{cor:LB-unbiased}. The authors are grateful to Mark Sellke and Chara Podimata for brief collaborations (with A. Slivkins) in the initial stages of this project. \vspace{2mm}
}}
  \newcommand{\country}[1]{#1.}
  \newcommand{\city}[1]{#1,}
  \newcommand{\institution}[1]{#1,}
  \newcommand{\email}[1]{Email: \texttt{#1}.}
  \newcommand{\affiliation}{\thanks}

\author{
 {Kiarash Banihashem
 \affiliation{
   \institution{University of Maryland at College Park}
   \city{College Park, MD}
   \country{USA}
	 {Supported by DARPA QuICC NSF and AF:Small \#2218678, \#2114269.}
 \email{\{kiarash,hajiagha,suhoshin\}@umd.edu}
 }}
 \and
 MohammadTaghi Hajiaghayi\footnotemark[2]
 \and
 Suho Shin\footnotemark[2]
 \and
 Aleksandrs Slivkins
 \affiliation{
   \institution{Microsoft Research NYC}
   \city{New York, NY}
   \country{USA}
 \email{slivkins@microsoft.com}
 }
}

\date{First version: February 2023\\This version: April 2025}

\begin{document}

\maketitle

\begin{abstract}
\input{abstract}

\end{abstract}

\newpage

\addtocontents{toc}{\protect\setcounter{tocdepth}{2}}
\tableofcontents

\newpage

\input{intro}

\input{sec-related}

\input{prelim}

\input{freq_LB}
\input{app-LB-extension}

\input{freq_UB}

\input{sec-Bayes}

\input{sec-fully-Bayesian}

\input{sec-K}

\input{sec-expts}

\newpage
\input{sec-conclusions}

% \input{freq_LB_proof.tex}
% Paper body

\newpage
\begin{small}
\bibliographystyle{plainnat}
\bibliography{ref,bib-abbrv,bib-bandits,bib-AGT,bib-slivkins}
\end{small}
\newpage

\appendix
\input{app-LB-tools}
\input{freq_UB_proof}
\input{app-bayesian}
%\input{app-BG}

% \newpage
% \section*{OLD STUFF, TO BE REMOVED}
% \newpage
% \input{old_intro}
% \input{old_freq_hetero}
% \input{old_bayes_homo}
% \input{old_in_progress.tex}
% \input{old_proofs.tex}
% \input{old_text.tex}
\end{document}

%% file: mymath.tex
\newcommand{\separator}{
  \begin{center}
    \rule{\columnwidth}{0.3mm}
  \end{center}
}

\newcommand{\beq}{\begin{eqnarray*}}
\newcommand{\eeq}{\end{eqnarray*}}
\newcommand{\beqn}{\begin{eqnarray}}
\newcommand{\eeqn}{\end{eqnarray}}
\newcommand{\bemn}{\begin{multiline}}
\newcommand{\eemn}{\end{multiline}}

\def\cM{{\mathcal M}}

\def\cA{{\mathcal A}}

\def\cH{{\mathcal H}}

\def\cE{{\mathcal E}}

\def\N{\mathbb{N}}
\def\R{\mathbb{R}}

\def\cP{{\mathcal P}}

\def\EE{{\mathbb E}}

%% file: abstract.tex
We study social learning dynamics motivated by reviews on online platforms. The agents collectively follow a simple multi-armed bandit protocol, but each agent acts myopically, without regards to exploration. We allow the greedy (exploitation-only) algorithm, as well as a wide range of behavioral biases. Specifically, we allow myopic behaviors that are consistent with (parameterized) confidence intervals for the arms' expected rewards. We derive stark learning failures for any such behavior, and provide matching positive results. The learning-failure results extend to Bayesian agents and Bayesian bandit environments.

In particular, we obtain general, quantitatively strong results on failure of the greedy bandit algorithm, both for ``frequentist" and ``Bayesian" versions. Failure results known previously are quantitatively weak, and either trivial or very specialized. Thus, we provide a theoretical foundation for designing non-trivial bandit algorithms, \ie algorithms that intentionally explore, which has been missing from the literature.

Our general behavioral model can be interpreted as agents' optimism or pessimism. The matching positive results entail a maximal allowed amount of optimism. Moreover, we find that no amount of pessimism helps against the learning failures, whereas even a small-but-constant fraction of extreme optimists avoids the failures and leads to near-optimal regret rates.

%thus providing a theoretical foundation for why bandit algorithms should explore.

%% file: intro.tex
% !TEX root =  main.tex
\section{Introduction}
\label{sec:intro}

% NB: introduce "social learning" in para1.

Reviews and ratings are pervasive in many online platforms. A customer consults reviews/ratings, then chooses a product and then (often) leaves feedback, which is aggregated by the platform and served to future customers. Collectively, customers face a tradeoff between \emph{exploration} and \emph{exploitation}, \ie between acquiring new information while making potentially suboptimal decisions and making optimal decisions using available information. However, individual customers tend to act myopically and favor exploitation, without regards to exploration for the sake of the others. Thus, we have a variant of social learning under exploration-exploitation tradeoff. On a high level, we ask \textbf{whether/how the myopic behavior interferes with efficient exploration}.
We are particularly interested in \emph{learning failures} when only a few agents choose an optimal action.%
\footnote{A weaker form of this phenomenon, when with positive probability an optimal action is chosen only finitely often under an infinite time horizon, is known as \emph{incomplete learning}.}

Taking a step back, exploration-exploitation tradeoff is fundamental in the study of sequential decision-making, and the central issue in the popular framework of \emph{multi-armed bandits}~\citep{slivkins-MABbook,LS19bandit-book}.
%
%While central to bandit algorithms,
However, intentional exploration can be problematic when an algorithm interacts with human users, as it imposes an arguably unfair burden on the current user for the sake of the future users.
Also, exploration adds complexity to algorithm/system design and necessitates substantial buy-in and engineering support for adoption in practice  \citep{MWT-WhitePaper-2016,DS-arxiv}.
The \emph{greedy algorithm}, which exploits known information at every step without any intentional exploration, sidesteps these issues and aligns well with customer incentives in the social-learning scenario described above.

The greedy algorithm is widely believed to perform poorly. Accordingly, the huge literature on multi-armed bandits overwhelmingly focuses on intentional exploration. A key motivation for this comes from the following simple argument. Consider Bernoulli $K$-armed bandits, where the reward of each of the $K$ arms follows an independent Bernoulli distribution with a fixed mean (and nothing else is known). Then the greedy algorithm, initialized with some samples of each arm, never tries the ``good" arm if all initial samples of this arm return $0$, and at least one sample of some other arm returns $1$. So, we have a \emph{learning failure} -- convergence on a wrong arm -- which happens with a positive-constant probability over the randomness in the initial samples.

However, we argue that our understanding of the greedy algorithm is very incomplete. Indeed, the failure probability in the above example is exponential in $N_0$, the number of initial samples. This is very weak for $N_0\gg 1$, and vacuous for $N_0\gg\log T$, where $T$ is the time horizon of interest. Other failure examples in the literature concern very specific one-dimensional linear structures, and do not characterize the failure probability.%
\footnote{See Related Work. In fact, the greedy algorithm is also known to perform well under some strong assumptions.}
Thus, while the greedy algorithm is believed to be inefficient in some strong and general sense, \emph{we do not know} whether this is the case, even for Bernoulli 2-armed bandits, and what should be the ``shape" of such results.%
%This is a remarkable gap in foundational knowledge for bandit theory!

Circling back to our social-learning scenario, we are interested in the greedy algorithm as well as a range of approximate-greedy behaviors. The greedy algorithm corresponds to rational behavior of the myopic customers, while the approximations thereof represent various \emph{behavioral biases} that the customers might have. The failure examples from prior work (however limited), do not apply at all to these approximate-greedy behaviors.

%its approximations correspond to various \emph{behavioral biases} that these customers might have.

%We are interested in the greedy algorithm itself as well as a range of approximately-greedy algorithms. Circling back to our social-learning scenario, the greedy algorithm represents the rational behavior of myopic customers. The approximately-greedy algorithms represent various \emph{behavioral biases} that myopic customers might have.

%Circling back to our social-learning scenario, we are interested in the greedy algorithm (representing customers' rational behavior), as well as a range of approximately-greedy algorithms representing various \emph{behavioral biases} that myopic customers might have.

\xhdr{Our model: Bandit Social Learning (\BSL).}
We distill the social-learning scenario
%We distill the tension between exploration and myopic behavior
down to its purest form, where the self-interested customers (henceforth, \emph{agents}) follow a simple multi-armed bandit protocol. The agents arrive sequentially and make one decision each. They
have two alternative products/experiences to choose from, termed \emph{arms}.%
\footnote{We focus on two arms unless specified otherwise; we consider $K>2$ arms in Section~\ref{sec:K}.}
%and no way to infer anything about one arm from the other.%
%\footnote{We also consider extensions with correlated Bayesian priors and $K\geq 2$ arms (Sections \ref{sec:priors} and \ref{sec:K}).}
Upon choosing an arm, the agent receives a reward: a Bernoulli random draw whose mean is specific to this arm and not known. The platform provides each agent with full history of the previous agents: the arms chosen and the rewards received.%
\footnote{In practice, online platforms provide summaries such as the average score and the number of samples.}
The agents do not observe any payoff-relevant signals prior to their decision, whether public or private.
%In particular, they believe they are similar to one another \sscomment{This sentence is not very-clear to me.}.

When all agents are governed by a centralized algorithm, this setting is known as \emph{stochastic bandits}, a standard and well-understood variant of multi-armed bandits. The greedy algorithm chooses an arm with a largest empirical reward in each round.

Initial knowledge is available to all agents: a dataset with $N_0$ samples of each arm. This knowledge represents reports created outside of our model, \eg by ghost shoppers, influencers, paid reviewers, journalists, etc., and available before (or soon after) the products enter the market. While the actual reports may have a different format, they shape agents' initial beliefs. So, one could interpret our initial data-points as a simple ``frequentist" representation for these initial beliefs. Accordingly, parameter $N_0$ determines the ``strength" of the beliefs.

We allow a wide range of myopic behaviors that are consistent with available observations. Consider standard upper/lower confidence bounds for the reward of each arm: the sample average plus/minus the ``confidence term" that scales as a square root of the number of samples.
%\sscomment{Would it be better if we say that at least upper confidence term is typical for efficiency in bandits somewhere before this sentence or in footnote? Seems confidence bound first appears hear and thought some might see it less motivated. Come to read next paras, it may be fine as-is.}.
Each agent evaluates each arm to an \emph{index}: some number that is consistent with these confidence bounds (but could be arbitrary otherwise), and chooses an arm with a largest index.%
\footnote{Whether the agents explicitly compute the confidence bounds is irrelevant to our model.}
The confidence term is parameterized by some factor $\sqrt{\eta}\geq 0$ to ensure that the true mean reward lies between the confidence bounds with probability at least $1-e^{-2\eta}$. We call such agents \emph{$\eta$-confident}. We emphasize that $\eta$ is a parameter of the model, rather than something that can be adjusted.

This model subsumes the ``unbiased" behavior, when the index equals the sample average (corresponding to the greedy algorithm), as well as ``optimism" and ``pessimism", when the index is, resp., larger or smaller than the sample average.%
\footnote{To set the notation, \emph{$\eta$-optimistic} (resp., \emph{$\eta$-pessimistic}) agents set their index to the respective upper (resp., lower) confidence bound parameterized by $\eta$. They are $\eta'$-confident for any $\eta'\geq\eta$.}
Such optimism/pessimism can also be interpreted as risk preferences. The index can be randomized, so the less preferred arm might still be chosen with some probability. 
%that the less preferred arm is chosen with a smaller, but strictly positive probability.
Further, an agent may be more optimistic about one arm than the other, and the amount of optimism / pessimism may depend on the previously observed rewards of either arm,
%and even favor more recent observations.  Finally,
and different agents may exhibit different behaviours within the permitted range. A more detailed discussion of the permitted behaviors can be found in Sections~\ref{sec:related} and~\ref{sec:prelims-cases}. 

%\newpage
We target the regime when parameter $\eta$  is a constant relative to $T$, the number of agents, \ie the agents' population is characterized by a constant $\eta$. %We are interested in the asymptotic behavior when $T$ increases.
An extreme version of our model, with $\eta\sim\log(T)$, is only considered for intuition and sanity checks. Interestingly, this extreme version subsumes two well-known bandit algorithms:
UCB1 \citep{bandits-ucb1} and Thompson Sampling
\citep{Thompson-1933,TS-survey-FTML18}, which achieve optimal regret bounds.%
\footnote{In particular, UCB1 is simply $\eta$-optimism with $\eta\sim\log T$.}
 These algorithms exemplify two standard design paradigms in bandits and reinforcement learning: resp., optimism under uncertainty and posterior sampling.
% -- which usefully extend to many scenarios in bandits and reinforcement learning.
They can also be seen as behaviors: resp., extreme optimism and \emph{probability matching} \citep{myers1976,vulkan2000}, a well-known randomized behavior. More ``moderate'' versions of these behaviors are consistent with $\eta$-confident agents.
% as defined above, and are subject to the learning failures described below.

%An extreme version of our model, with $\eta\sim\log(T)$, subsumes two well-known bandit algorithms. UCB1 algorithm \cite{bandits-ucb1} makes the index equal to the respective upper confidence bound. Thompson Sampling \cite{Thompson-1933,TS-survey-FTML18} draws the index of each arm independently from the corresponding Bayesian posterior. It can be seen as a variant of \emph{probability matching} \citep{myers1976,vulkan2000}, a well-known randomized behavior. Both algorithms achieve regret that scales as $\log(T)$ for a particular problem instance and as $\sqrt{T \log T}$ in the worst case, both of which are essentially optimal across all bandit algorithms.%
%\footnote{Here and elsewhere, regret is defined as the expected difference in expected total reward between the best arm and the algorithm/agents. It is a very standard performance measure in online machine learning.}
%The two algorithms exemplify standard design paradigms -- resp., optimism under uncertainty and posterior sampling -- which usefully extend to many scenarios in bandits and reinforcement learning. More ``moderate" versions of these behaviors are consistent with $\eta$-confidence as defined above, and are subject to the learning failures described below.

\begin{table}[t]
\centering
\begin{tabular}{l|l|l}
\multicolumn{1}{c|}{Direction}
    & \multicolumn{1}{c|}{Behavior}
    & \multicolumn{1}{c}{Results}
\\[.5mm]\hline &&\\[-1.5ex]
negative& $\eta$-confident      & Thm.~\ref{thm:interval} (main),\\
        &                       & Thm.~\ref{thm:LB-ext} (small $N_0$). \\
        & unbiased/Greedy       & Cor.~\ref{cor:LB-unbiased} \\
        & $\eta_t$-pessimistic  & Thm.~\ref{thm:pessimis}
\\[.5mm]\hline &&\\[-1.5ex]
positive& $\eta$-optimistic     & Thm.~\ref{thm:upper} \\
        & $\eta_t$-optimistic, $\eta_t\in [\eta,\etamax]$
                                & Thm.~\ref{thm:upper-interval} \\
        & small fraction of optimists & Thm.~\ref{thm:upper-distr}.
\end{tabular}
\caption{Our results for frequentist agents.}
%\vspace{-2mm}
\label{tab:results-freq}
\end{table}

%\begin{table}[t]
%\centering
%\begin{tabular}{l|l|l|l}
%\multicolumn{1}{c|}{Mean rewards }
%        & \multicolumn{1}{c|}{Beliefs}           & \multicolumn{1}{c|}{Behavior}  & \multicolumn{1}{c}{Results}  \\[.5mm]\hline &&&\\[-1.5ex]
%fixed   & ``frequentist"
%                            & $\eta$-confident      & Thm.~\ref{thm:interval} (main),\\
%        & ~~~~~confidence intervals
%                            &                       & Thm.~\ref{thm:LB-ext} (small $N_0$). \\
%        &                   & unbiased/Greedy       & Cor.~\ref{cor:LB-unbiased} \\
%        &                   & $\eta_t$-pessimistic  & Thm.~\ref{thm:pessimis} \\
%        & Bayesian (independent)
%                            & Bayesian-unbiased        & Thm.~\ref{thm:Bayes-unbiased}(a) \\
%        &                   & $\eta$-Bayesian-confident& Thm.~\ref{thm:Bayes-unbiased}(b) \\
%Bayesian (correlated)
%        & Bayesian (and correct)
%                            & Bayesian-unbiased        & Thm.~\ref{thm:BG};
%                                                            Cor.~\ref{cor:BG},~\ref{cor:BG-general}
%\end{tabular}
%\caption{Our negative results: learning failures.}
%%\vspace{-2mm}
%\label{tab:results-negative}
%\end{table}

\xhdr{Our results.}
We are interested in \emph{learning failures} when all but a few agents choose the bad arm, how the failure probability scales with the relevant parameters, and how the resulting regret rates scale as a function of $T$, the number of agents.

Our first result concerns unbiased agents ($\eta=0$), \ie the greedy algorithm. We obtain failure probability $\pfail = \Omega(1/\sqrt{N_0})$, where $N_0$ is the number of initial samples. This is an exponential improvement over the trivial argument presented above.%
\footnote{This result holds as long as the \emph{gap} $\Delta$ (the absolute difference in arms' mean rewards) is smaller than $1/\sqrt{N_0}$. This is the only non-trivial regime: indeed, when $\Delta \gg 1/\sqrt{N_0}$, one could infer the best arm with high confidence based on the initial samples alone.}
Regret is at least $\Omega(\pfail\cdot T)$ for any given problem instance, in contrast with the $O(\log T)$ regret rate obtained by optimal bandit algorithms.

Our main results concern $\eta$-confident agents and investigate the scaling in $\eta$. We obtain failure probability $\pfail = e^{-O(\eta)}$ (with a similar scaling in $N_0$), specializing to the greedy algorithm when $\eta=0$ (see Section~\ref{sec:fail}). Further, the
$e^{-O(\eta)}$ scaling is the strongest possible: indeed, regret for $\eta$-optimistic agents is at most
    $O\rbr{T\cdot e^{-\Omega(\eta)} + \eta}$
for a given problem instance (Theorem~\ref{thm:upper}). Note that the negative result deteriorates as $\eta$ increases, and  becomes vacuous when $\eta\sim\log T$. Then $\eta$-optimistic agents correspond to the UCB1 algorithm \citep{bandits-ucb1}, and our upper bound essentially matches its optimal $O(\log T)$ regret rate.

We refine these results in several directions.
First, \emph{pessimism does not help}: if all agents are pessimistic, then any level of pessimism, whether small or large or different across agents, leads to at least the same failure probability as in the unbiased case (Theorem~\ref{thm:pessimis}).
Second, our positive result for $\eta$-confidence agents is robust, in the sense that some agents can be \emph{more} optimistic (Theorem~\ref{thm:upper-interval}).
Third, \emph{a small fraction of optimists goes a long way}! Namely, if all agents are $\eta$-confident and even a $q$-fraction of them are $\eta$-optimistic, this yields regret
        $O\rbr{T\cdot e^{-\Omega(\eta)} + \eta/q}$,
almost as if \emph{all} agents were $\eta$-optimistic.%
\footnote{A similar result holds even the agents hold different levels of optimism, \eg if each agent $t$ in the $q$-fraction is $\eta$-optimistic for some $\eta_t\geq \eta$. See Theorem~\ref{thm:upper-distr} for the most general formulation.}
All results are summarized in Table~\ref{tab:results-freq}.

We provide numerical simulations to illustrate our key findings (Section~\ref{sec:expts}). In these simulations, we investigate the probability of a learning failure for a particular bandit instance and a particular ``behavior type" (expressed by the $\eta$ parameter). Specifically, we plot the probability of never choosing the good arm after some round $t$, as a function of $t$. We find substantial learning failures, as predicted by the theory, which subside for $\eta$-optimistic agents as $\eta$ increases.

\xhdr{Bayesian agents.}
We also consider agents who have Bayesian beliefs and act according to their posteriors given the observed data (henceforth, \emph{Bayesian agents}). This is in contrast with purely data-driven agents in our main model, as described previously; to make a distinction, we will refer to the latter as \emph{frequentist} agents. We posit that the Bayesian beliefs are same for all agents, representing the common initial knowledge, along with the initial data points. The beliefs are independent across arms, unless specified otherwise; then, like for our frequentist agents, observations from one arm do not yield information about the other arm.

A rational behavior in this Bayesian setup is to choose an arm with a largest posterior mean reward. This behavior, called  \emph{Bayesian-unbiased}, can be seen as a Bayesian version of the greedy algorithm. It is believed to perform poorly, like its frequentist counterpart. While a trivial argument yields a learning failure for deterministic rewards,%
\footnote{Letting $\mu_1>\mu_2$ be the arms' rewards, suppose
    $\Ex{\mu_1}<\min\rbr{\mu_2,\Ex{\mu_2}}$,
where the expectation is over the beliefs (conditioned on the initial data). Then the Bayesian greedy algorithm always chooses arm $2$.}
we are not aware of \emph{any} negative results when the rewards are  randomized.

We consider Bayesian agents on a fixed bandit instance (Section~\ref{sec:bayes}). We observe that Bayesian-unbiased agents
%and expressed by Beta distributions, such agents
are consistent with frequentist $\eta$-confident agents, for some $\eta$ determined by $N_0$ and the beliefs, and therefore are subject to the same negative results. Moreover, we define a Bayesian version of $\eta$-confident agents, with confidence intervals determined by the posterior, and show that such agents are consistent with frequentist $\eta'$-confident agents for an appropriate $\eta'$.

We also consider a ``fully Bayesian" model in which the arms' mean rewards $(\mu_1,\mu_2)$ are drawn from a common Bayesian prior $\cP$ (Section~\ref{sec:priors}). Put differently, Bayesian agents operate in the environment of Bayesian bandits, and both are driven by the same prior $\cP$. For this model, we focus on the paradigmatic case of Bayesian-unbiased agents and no initial data. We derive a negative result for an arbitrary prior: if arm $1$ is preferred according to the prior, then the probability of never choosing arm $2$ is at least $\EXP_\cP[\mu_1-\mu_2]$. This yields a learning failure when arm $2$ is in fact the best arm; we characterize the probability of this happening in terms of the prior. In fact, this result extends to priors that are correlated across arms, albeit with a substantial caveat: the failure probability is driven by the minimal probability density across all pairs $(\mu_1,\mu_2)\in[0,1]^2$.
\footnote{Put differently, we need a full-support assumption on the prior: every pair $(\mu_1,\mu_2)\in[0,1]^2$ occurs with probability density at least $\minPDF>0$, and the probability of a learning failure is driven by $\minPDF$.}

\xhdr{Extensions to $K>2$ arms.}
We extend our negative results to $K\geq 2$ arms, for both frequentist and Bayesian agents (Section~\ref{sec:K}). We establish learning failures for any given problem instance, in a similar sense as in the respective $K=2$ cases. These extensions use essentially the same proof techniques, but require somewhat more complex formulations. \Eg the frequentist result considers the probability of never choosing any of the top $m$ arms, and characterizes it in terms of the gap between the best and the $n$-th best arm, for any given $m<n$.

\xhdr{Discussion: significance.}
%We target the scenario in social learning when both actions and rewards are observable in the future, and the agents do not receive any other payoff-relevant signals.
Our goal is to analyze the intrinsic learning behavior of a system of self-interested agents, rather than design a new algorithm/mechanism for such system. As in much of algorithmic game theory, we discuss the influence of self-interested behavior on the overall welfare of the system. We consider ``learning failures'' caused by self-interested behavior, which is a typical framing in the literature on social learning.

While our positive results are restricted to ``optimistic" agents, we do not assert that such agents are necessarily typical. Instead, we establish that our results on learning failures are essentially tight. That said, ``optimism" is a well-documented behavioral bias \citep[\eg see][]{Puri-2007}. So, a small fraction of optimists, leveraged in Theorem~\ref{thm:upper-distr}, is not unrealistic.

From the algorithmic perspective, we showcase the failures of the greedy algorithm, and more generally any algorithm that operates on narrow confidence intervals. We do not attempt to design new algorithms for $K$-armed bandits, as several optimal algorithms are already known. As a by-product, our results on $\eta$-confident agents elucidate some important aspects of \emph{exploration}: why bandit algorithms require (some) extreme optimism --- to be inconsistent with $\eta$-confident agents for a constant $\eta$ --- and why ``pessimism under uncertainty" is not a productive approach.

\xhdr{Technical novelty.}
Bandit Social Learning was not well-understood previously even with unbiased agents, as discussed above, let alone for more permissive behavioral models. It was very unclear a priori how to analyze learning failures and how strong would be the guarantees, in terms of the generality of agents' behaviors, the failure events/probabilities, and the technical assumptions.

On a technical level, our ``negative" proofs have very little (if anything) to do with standard lower-bound analyses in bandits stemming from \citet{Lai-Robbins-85} and \citet{bandits-exp3}. These analyses apply to any algorithm and prove ``sublinear" lower bounds on regret, such as $\Omega(\log T)$ for a given problem instance and  $\Omega(\sqrt{T})$ in the worst case. On a technical level, they present a KL-divergence argument showing that no algorithm can distinguish between a given tuple of "similar" problem instances. In contrast, we prove \emph{linear} lower bounds on regret, our results apply to a particular family of behaviors/algorithms, and we never consider a tuple of similar problem instances. Instead, we use  anti-concentration and martingale tools to argue that the best arm is never played (or played only a few times), with some probability.
While our tools themselves are not very standard, the novelty is primarily in how we use these tools.
%The result on correlated beliefs in Section~\ref{sec:priors} has a rather short but "conceptual" proof which we believe is well-suited for a textbook.

Our ``positive" proofs are more involved compared to the standard analysis of the UCB1 algorithm. The latter uses $\eta\sim\log T$
%This is because we cannot make the $\eta$ parameter as large as needed
to ensure that the complements of certain ``clean events" can be ignored. Instead, we need to define and analyze these ``clean events" in a more careful way. These difficulties are compounded in Theorem~\ref{thm:upper-distr}, our most general result. As far as the statements are concerned, the basic result in Theorem~\ref{thm:upper} is perhaps what one would expect to hold, whereas the extensions in Theorem~\ref{thm:upper-interval} and~\ref{thm:upper-distr} are more surprising.

%As far as the theorem statements are concerned, the basic result in Theorem~\ref{thm:upper} is essentially as expected, whereas the extensions in Theorem~\ref{thm:upper-interval} and Theorem~\ref{thm:upper-distr} are less so.

%\xhdr{Discussion: framing.}
%Our paper is primarily on social learning, rather than on bandit algorithms. Put differently, the primary goal is to analyze the learning behavior of a system of agents. We make connections to bandit algorithms so as to serve this goal, but we do not attempt to design algorithms based on these connections. Meanwhile, implications for the greedy bandit algorithm are of independent interest, as discussed above.
%
%Within social learning, we target the scenario when both actions and rewards are observable in the future. This scenario is motivated by reviews and ratings on online platforms and is studied in some prior work. In Related Work, we discuss these connections in more detail and separate our work from other strands of social learning.
%
%More abstractly, we ask how self-interested behavior of the individuals impacts the aggregate welfare of the system, compared to the case when the agents' behavior is controlled by an optimal algorithm. This is a key question in algorithmic game theory, studied in many different scenarios, usually under the framing of the ``price of anarchy". In contrast, we talk about ``learning failures", a typical framing from the literature on social learning.

\xhdr{Map of the paper.}
Section~\ref{sec:prelims} introduces our model in detail and discusses various allowed behaviors. Sections~\ref{sec:fail} and~\ref{sec:UB} discuss, resp., the learning failures and the positive results for our main (frequentist) model.
% derives the learning failures. Section~\ref{sec:UB} presents our positive results.
Sections~\ref{sec:bayes} and~\ref{sec:priors} handle Bayesian agents: resp., for a fixed (frequentist) bandit instance and for Bayesian bandits. Negative results for $K\geq 2$ arms are Section~\ref{sec:K}. Numerical simulations are in Section~\ref{sec:expts}. Some unessential proofs are moved to appendices.

%% file: sec-related.tex
\section{Related Work}
\label{sec:related}

\noindent\textbf{Social learning.}
A vast literature on social learning studies agents that learn over time in a shared environment. Learning failures such as ours (or absence thereof) is a prominent topic. Models vary across several dimensions, such as: which information is acquired or transmitted, what is the communication network, whether agents are long-lived or only act once, how they choose their actions, etc. All models from prior work are very different from ours. Below we separate our model from several lines of work that are most relevant.

In ``sequential social learning", starting from \citep{Banerjee-qje92,Welch92,Bikhchandani-jpe92,SmithSorensen-econometrica00},
agents observe private signals, but only the chosen actions are observable in the future; see \citet{Golub-survey16} for a survey. The social planner (who chooses agents' actions given access to the knowledge of all previous agents) only needs to \emph{exploit}, \ie choose the best action given the previous agents' signals, whereas in our model it also needs to \emph{explore}. Learning failures are (also) of primary interest, but they occur for an entirely different reason: restricted information flow, since the private signals are not observable in the future.

``Strategic experimentation'', starting from \citet{Bolton-econometrica99} and \citet{Keller-econometrica05}, studies long-lived learning agents that observe both actions and rewards of one another; see \citet{Horner-survey16} for a survey. Here, the social planner also solves a version of multi-armed bandits, albeit a very different one (with time-discounting, ``safe" arm that is completely known, and ``risky" arm that follows a stochastic process). The main difference is that the agents engage in a complex repeated game where they explore but prefer to free-ride on exploration by others.

\citet{Bala-Goyal-98} and \citet{Lazer-ASQ07} consider a network of myopic learners, all  faced with the same bandit problem and observing each other's actions and rewards. The interaction protocol is very different from ours: agents are long-lived, act all at once, and only observe their neighbors on the network. Other specifics are different, too. \citet{Bala-Goyal-98} makes strong assumptions on learners' beliefs, which would essentially cause the greedy algorithm to work well in \BSL. In \citet{Lazer-ASQ07}, each learner only retains the best observed action, rather than the full history. The focus is on comparing the impact of different network topologies, theoretically  \citep{Bala-Goyal-98} and via simulations \citep{Lazer-ASQ07}.

Prominent recent work, \eg \citep{Strak-ecta18,Bohren-ecta21,Lanzani-ecta21,Lanzani-jmp}, targets agents with \emph{misspecified beliefs}, \ie beliefs whose support does not include the correct model. The framing is similar to \BSL with Bayesian-unbiased agents: agents arrive one by one and face the same decision problem, whereby each agent makes a rational decision after observing the outcomes of the previous agents.%
\footnote{This work usually posits a single learner that makes (possibly) myopic decisions over time and observes their outcomes. An alternative interpretation is that each decision is made by a new myopic agent who observes the history.}
Rational decisions under misspecified beliefs make a big difference compared to \BSL, and structural assumptions about rewards/observations and the state space tend to be very different from ours. The technical questions being asked tend to be different, too. \Eg convergence of beliefs is of primary interest, whereas the chosen arms and agents' beliefs/estimates trivially converge in our setting.
\footnote{Essentially, if an arm is chosen infinitely often then the agents beliefs/estimates converge on its true mean reward; else, the agents eventually stop receiving any new information about this arm.}

%A prominent line of work on \emph{endogenous learning}, \eg \citep{Strak-ecta18,Bohren-ecta21,Lanzani-ecta21,Lanzani-jmp} focuses on a single learner that makes (possibly) myopic decisions over time and observes their outcomes. Our model admits a similar interpretation, too (with all agents having the same behavioral type). However, this line of work focuses on a specific phenomenon of \emph{misspecified beliefs} (\ie beliefs whose support does not include the correct model), posits that the learner is rational under these beliefs, and makes structural assumptions that are very different from ours. The technical questions being asked tend to be different, too, \eg convergence of beliefs is of primary interest.%
%\footnote{In contrast, the chosen arms and agents' beliefs/estimates trivially converge in our setting. Essentially, if an arm is chosen infinitely often then the agents beliefs/estimates converge on its true mean reward; else, the agents eventually stop receiving any new information about this arm.}

\xhdr{The greedy algorithm.}
Positive results for the greedy bandit algorithm focus on \emph{contextual bandits}, an extension of stochastic bandits where a payoff-relevant signal (\emph{context}) is available before each round. Equivalently, this is a version of \BSL with Bayesian-unbiased agents where each agent observes an idiosyncratic signal along with the history, which is visible to the future agents. The greedy algorithm has been proved to work well under very strong assumptions on the environment: linearity of rewards and diversity of contexts
\citep{kannan2018smoothed,bastani2017exploiting,Greedy-Manish-18}. Similarly, \citet{AcemogluMMO19} analyze \BSL with Bayesian-unbiased agents who receive private idiosyncratic signals. They  make (different) strong assumptions on agent diversity and reward structure, and focus on \emph{one-armed bandits} (when the alternative is ``do nothing", concretely: buy the product or not).%
\footnote{\citet{AcemogluMMO19} also obtain complementary results on the existence of learning failures in their setting; quantitatively, these negative results similar to the exponentially-weak failure discussed in Section~\ref{sec:intro}. They further  zoom in on the effects of biased reporting and summarized history, which goes beyond our scope here.}
In all these results, context/agent diversity substitutes for exploration, and reward structure allows aggregation across agents.

The greedy algorithm is also known to attain $o(T)$ regret in various scenarios with a very large number of near-optimal arms \citep{Bayati-nips20,Jedor-Greedy21},
%Particularly, it obtains $o(T)$ regret in various scenarios that ensure a very large number of \emph{near-optimal} arms.
\eg for Bayesian bandits with $\gg\sqrt{T}$ arms, where the arms' mean rewards are sampled independently and uniformly.

Learning failures for the greedy algorithm are derived for bandit problems with 1-dimensional action spaces under (strong) structural assumptions: \eg dynamic pricing with linear demands \citep{Zeevi-ManSci12,denBoer-ManSci14} and dynamic control in a (generalized) linear model \citep{Lai-Robbins-control82,Keskin-OpRe18}. In all these results, the failure probability is only proved positive, but not otherwise characterized. The greedy algorithm is restricted to one or two initial samples (which is a trivial case in our setting, as discussed in Section~\ref{sec:intro}).

\xhdr{\BSL and mechanism design.}
\emph{Incentivized exploration} takes a mechanism design perspective on \BSL, whereby the platform strives to incentivize individual agents to explore for the sake of the common good. In most of this work, starting from \citep{Kremer-JPE14,Che-13},
the platform controls the information flow, \eg can withhold history and instead issue recommendations, and uses this information asymmetry to create incentives; surveys can be found in \citep{IncentivizedExploration-chapter} and \citep[Ch. 11]{slivkins-MABbook}. In particular, \citep{ICexploration-ec15,Jieming-unbiased18,Selke-PoIE-ec21} target stochastic bandits as the underlying learning problem, same as we do. Most related is \citet{Jieming-unbiased18}, where the platform constructs a (very) particular communication network for the agents, and then the agents engage in \BSL on this network.

Alternatively, the agents are allowed to observe full history, but the platform uses monetary payments to create incentives \citep{Frazier-ec14,Kempe-wine15,Kempe-colt18}. The platform's goal is to optimize the welfare vs. payments tradeoff under time-discounting.

\xhdr{Behaviorial models.}
Non-Bayesian models of behavior are prominent in social learning literature, starting from \citet{DeGroot74}. In these models, agents use variants of statistical inference and/or naive rules-of-thumb to infer the state of the world from observations. Our model of $\eta$-confident agents is essentially a special case of ``case-based decision theory" of \citet{CaseBased-qje95}.

Our model accommodates versions of several behaviorial biases:
\begin{OneLiners}
\item optimism (\eg see \citep{Puri-2007} and references therein),
\item pessimism (\eg see \citep{Chang00-book,Bateson-2016} and references therein),
\item risk attitudes (\eg see \citep{KahnemanTversky-book82,Barberis-survey03}),
\item recency bias (\eg see \citep{Fudenberg2014} and references therein),
\item randomized decisions (with theory tracing back to \citet{Luce59}), and
\item \emph{probability matching} more specifically (\eg see surveys \citep{myers1976,vulkan2000}).
\end{OneLiners}
All these biases are well-documented and well-studied in the literature on economics and psychology. A technical discussion of how these and other behaviors fit into our model is in Section~\ref{sec:prelims-cases}.

%probab matching: \citep{agarwal2008learning,ErevHaruvy-2017,Fudenberg2014}

\xhdr{Multi-armed bandits.}
Our perspective on bandits is very standard in machine learning theory: we consider asymptotic regret rates without time-discounting (rather than Bayesian-optimal time-discounted rewards, a more standard economic perspective). The vast literature on regret-minimizing bandits is summarized in recent books
\citep{slivkins-MABbook,LS19bandit-book}.

Stochastic bandits is a standard, basic version with i.i.d. rewards and no auxiliary structure. Most relevant are the UCB1 algorithm~\citep{bandits-ucb1}, Thompson Sampling and the ``frequentist" analyses thereof  \citep{Thompson-1933,TS-survey-FTML18,Shipra-colt12,Shipra-aistats13-JACM,Kaufmann-alt12}, and the lower bounds \cite[\eg][]{Lai-Robbins-85,bandits-exp3}.
The general design paradigms associated with UCB1 and Thompson Sampling are surveyed in \citep{slivkins-MABbook,LS19bandit-book,TS-survey-FTML18}.
%The lower bounds on regret of arbitrary algorithms   \cite[\eg][]{Lai-Robbins-85,bandits-exp3} are very different from our negative results, as explained in the Introducton.

Markovian, time-discounted bandit formulations \citep{Gittins-book11} and various other connections between bandits and self-interested behavior (surveyed, \eg in \citet[Chapter 11.7]{slivkins-MABbook}) are less relevant to this paper.

%% file: prelim.tex
% !TEX root =  main.tex
\section{Our model and preliminaries}
\label{sec:prelims}

Our model, called \textbf{Bandit Social Learning}, is defined as follows. There are $T$ rounds, where $T\in\N$ is the time horizon, and two \emph{arms} (\ie alternative actions). We use $[T]$ and $[2]$ to denote the set of rounds and arms, respectively.%
\footnote{Throughout, we denote $[n] = \cbr{1,2 \LDOTS n}$, for any $n\in\N$.}
In each round $t\in [T]$, a new agent arrives, observes history $\hist_t$ (defined below), chooses an arm $a_t\in [2]$, receives reward $r_t\in [0,1]$ for this arm, and leaves forever. When a given arm $a\in[2]$ is chosen, its reward is drawn independently from Bernoulli distribution with mean $\mu_a\in [0,1]$.
\footnote{Our results on upper bounds (Section~\ref{sec:UB}) and Bayesian learning failures (Section~\ref{sec:priors}) allow each arm to have an arbitrary reward distribution on $[0,1]$. We omit further mention of this to simplify presentation.}
The mean reward is fixed over time, but not known to the agents. Some initial data is available to all agents, namely $N_0\geq 1$ samples of each arm $a\in[2]$. We denote them $r^0_{a,i}\in[0,1]$, $i\in[N_0]$. The history in round $t$ consists of both the initial data and the data generated by the previous agents. Formally, it is a tuple of arm-reward pairs,
\[ \hist_t
    :=  \rbr{ (a,r^0_{a,i}):\; a\in[2], i\in [N_0] ;\;
        (a_s,r_s):\; s\in [t-1] }.\]
We summarize the protocol for Bandit Social Learning as Protocol~\ref{alg:BSL}.

\renewcommand{\algorithmcfname}{Protocol}
\begin{algorithm}
  \SetAlgoLined
  Problem instance: two arms $a\in[2]$ with (fixed, but unknown) mean rewards $\mu_1,\mu_2\in[0,1]$ \;
  Initialization:
  $\hist \leftarrow \cbr{\text{$N_0$ samples of each arm}} $\;
  \For{each round $t=1,2,\ldots, T$}{
    agent $t$ arrives, observes $\hist$ and chooses an arm $a_t\in[2]$ \;
    reward $r_t\in\sbr{0,1}$ is drawn from Bernoulli distribution
        with mean $\mu_{a_t}$\;
   new datapoint $(a_t,r_t)$ is added to $\hist$
  }
 \caption{Bandit Social Learning}\label{alg:BSL}
\end{algorithm}

\begin{remark}
%\ascomment{redundant given Intro/model, remove}
%The initial data-points represent reports created outside of our model, \eg by ghost shoppers, influencers, paid reviewers, journalists, etc., and available before (or soon after) the products enter the market. While the actual reports may have a different format, they shape agents' initial beliefs. So, one could interpret our initial data-points as a simple ``frequentist" representation for the initial beliefs.
The initial data points represent agents' initial beliefs; parameter $N_0$ determines the ``strength" of the beliefs. We posit $N_0\geq 1$ to ensure well-defined average rewards.
\end{remark}

If the agents were controlled by an algorithm, this protocol would correspond to  \emph{stochastic bandits} with two arms, the most basic version of multi-armed bandits. A standard performance measure in multi-armed bandits (and online machine learning more generally) is \emph{regret}, defined as
\begin{align}
  \regret(T) \defeq \textstyle \mu^*\cdot T - \Ex{\sum_{t\in[T]}\; \mu_{a_t}},
  \label{eq:regret}
\end{align}
where $\mu^* = \max(\mu_1,\mu_2)$ is the maximal expected reward of an arm.

Each agent $t$ chooses its arm $a_t$ myopically, without regard to future agents. Each agent is endowed with some (possibly randomized) mapping from histories to arms, and chooses an arm accordingly. This mapping, called \emph{behavioral type}, encapsulates how the agent resolves uncertainty on the rewards.  More concretely, each agent maps the observed history $\hist_t$ to an \emph{index}
    $\indx_{a,t}\in \R$
for each arm $a\in[2]$, and chooses an arm with a largest index. The ties are broken independently and uniformly at random.

We allow for a range of myopic behaviors, whereby each index can take an arbitrary value in the (parameterized) confidence interval for the corresponding arm. Formally, fix arm $a\in[2]$ and round $t\in [T]$. Let $n_{a, t}$ denote the number of times this arm has been chosen in the history $\hist_t$ (including the initial data), and let $\muhat_{a, t}$ denote the corresponding average reward. Given these samples, standard (frequentist, truncated) upper and lower confidence bounds for the arm's mean reward $\mu_a$ (UCB and LCB, for short) are defined as follows:
\begin{align}  \label{eq:def_ucb}
\ucb_{a, t} := \min\cbr{1, \muhat_{a, t} + \sqrt{\eta/n_{a, t}}}
\quad\text{and}\quad
\lcb_{a, t} := \max\cbr{0, \muhat_{a, t} - \sqrt{\eta/n_{a, t}}},
\end{align}
where $\eta\geq 0$ is a parameter. The interval
    $\sbr{\lcb_{a, t}, \ucb_{a, t}}$
will be referred to as \emph{$\eta$-confidence interval}. Standard concentration inequalities imply that $\mu_a$ is contained in this interval with probability at least $1-2\,e^{-2\eta}$ (where the probability is over the random rewards, for any fixed value of $\mu_a$). We allow the index to take an arbitrary value in this interval:
\begin{align}\label{eq:ind-conf}
 \indx_{a,t} \in \sbr{\lcb_{a, t}, \ucb_{a, t}},
    \quad \text{for each arm $a\in[2]$}.
\end{align}
We refer to such agents as \emph{$\eta$-confident}; $\eta>0$ will be a crucial parameter throughout.

We posit that the agents come from some population characterized by some fixed $\eta$, while the number of agents ($T$) can grow arbitrarily large. Thus, we are mainly interested in the regime when $\eta$ is a \emph{constant} with respect to $T$.

\subsection{Special cases of our model}
\label{sec:prelims-cases}

We emphasize the following special cases of $\eta$-confident agents:
\begin{itemize}

\item \emph{unbiased agents} set each index to the respective sample average:
$\indx_{a,t} = \muhat_{a, t}$. This is a natural myopic behavior for a ``frequentist" agent in the absence of behavioral biases.

\item \emph{$\eta$-optimistic agents} evaluate the uncertainty on each arm in the optimistic way, setting the index to the corresponding UCB: $\indx_{a,t} = \ucb_{a, t}$.

\item \emph{$\eta$-pessimistic agents} exhibit pessimism, in the same sense:
$\indx_{a,t} = \lcb_{a, t}$.

\end{itemize}

Unbiased agents correspond precisely to the \emph{greedy algorithm} in multi-armed bandits which is entirely driven by exploitation, and chooses arms as
    $a_t \in \argmax_{a\in[2]} \muhat_{a,t}$.
%This algorithm is known to be inefficient in some special cases, both empirically and theoretically.
In contrast, $\eta$-optimistic agents with $\eta\sim\log T$ correspond to UCB1 \citep{bandits-ucb1}, a standard algorithm for stochastic bandits which achieves optimal regret rates. We interpret such agents as exhibiting \emph{extreme} optimism, in that $\indx_{a,t} \geq \mu_a$ with very high probability. Meanwhile, our model focuses on (more) moderate amounts of optimism, whereby $\eta$ is a constant with respect to $T$.

\xhdr{Other behavioral biases.}
One possible interpretation for $\indx_{a,t}$ is that it can be seen as \emph{certainty equivalent}, \ie the smallest reward that agent $t$ is willing to take for sure instead of choosing arm $a$. Then $\eta$-optimism and $\eta$-pessimism corresponds to (moderate) \emph{risk-seeking} and \emph{risk-aversion}, respectively. In particular, $\eta$-pessimistic agents may be quite common.

Our model also accommodates a version of \emph{recency bias}, whereby recent observations are given more weight. For example, an $\eta$-confident agent may be $\eta$-optimistic for a given arm if more recent rewards from this arm are better than the earlier ones.

An $\eta$-confident agent could have a preference towards a given arm $a$, and therefore, \eg be $\eta$-optimistic for this arm and $\eta$-pessimistic for the other arm. The agent's ``attitude" towards arm $a$ could also be influenced by the rewards of the other arm, \eg (s)he could be $\eta$-optimistic for arm $a$ if the rewards from the other arms are high.

\xhdr{Randomized agents.}
Our model also accommodates \emph{randomized} $\eta$-confident agents, \ie ones that draw their indices from some distribution conditional on the history $\hist_t$. Such randomization is consistent with a well-known type of behaviors when human agents choose a seemingly inferior alternative with smaller but non-zero probability.

A notable special case is related to \emph{probability matching}, when the probability of choosing an arm equals to the (perceived) probability of this arm being the best. We formalize this case in a Bayesian framework, whereby all agents have a Bayesian prior such that the mean reward $\mu_a$ for each arm $a$ is drawn independently from the uniform distribution over $[0,1]$.
\footnote{This Bayesian prior is just a formality to define probability matching, not (necessarily) what the agents believe.}
Each agent $t$ computes the Bayesian posterior $\cP_{a,t}$ on $\mu_a$ given the history $\hist_t$, then samples a number $\nu_{a,t}$ independently from this posterior. Finally, we define each index $\indx_{a,t}$, $a\in[2]$ as the ``projection" of $\nu_{a,t}$ into the corresponding $\eta$-confidence interval
    $\sbr{\lcb_{a, t}, \ucb_{a, t}}$.
Here, the projection of a number $x$ into an interval $[a,b]$ is defined as $a$ if $x<a$, $b$ if $x>b$, and $x$ otherwise.

Here's why this construction is interesting. Without truncation, \ie when $\indx_{a,t} = \nu_{a,t}$,  each arm is chosen precisely with probability of this arm being the best according to the posterior
    $\rbr{\cP_{1,t},\,\cP_{2,t}}$.
In fact, this behavior precisely corresponds to \emph{Thompson Sampling} \citep{Thompson-1933},
another standard multi-armed bandit algorithm that attains optimal regret. For $\eta\sim\log T$, the system of agents behaves like Thompson Sampling with very high probability;%
\footnote{More formally:
    $\Pr{ \nu_{a,t}\in \sbr{\lcb_{a, t}, \ucb_{a, t}}: a\in[2], t\in[T]}>1-O(\nicefrac{1}{T})$,
if $\eta$ is large enough.}
we interpret such behavior as an \emph{extreme} version of probability matching. Meanwhile, we focus on moderate regimes such that $\eta$ is a constant with respect to $T$. We refer to such agents as \emph{$\eta$-Thompson agents}.

Let us flag two other randomized behaviors allowed by our model. First, a naive form of probability matching chooses an index of each arm independently and uniformly at random from the respective $\eta$-confidence interval. This is one way to express complete uncertainty on which values within each confidence interval are more likely. Second, an even more naive decision rule chooses an arm uniformly at random if the two $\eta$-confidence intervals overlap.%
\footnote{And if they don't, the arm with the higher interval must be chosen. Formally, the u.a.r. choice can be modeled via a correlated choice of the two indices, randomizing between (high,low) and (low, high).}
%Incidentally, this is \emph{active arms elimination} \citep{EvenDar-icml06}, a well-known (and regret-optimal) algorithm for multi-armed bandits.
Both behaviors provide stylized reference points for how ``naive" human agents may behave in practice.

%\xhdr{Bayesian agents.}
%We also accommodate agents that preprocess the observed data to a Bayesian posterior, and use the latter to define their indices; we term them \emph{Bayesian agents}.%
%\footnote{As opposed to ``frequentist" agents who preprocess the observed data to confidence intervals such as \eqref{eq:def_ucb}.}
%We analyze Bayesian versions of unbiased agents and $\eta$-confident agents, interpreting them as (frequentist) $\eta'$-confident agents defined above (with slightly larger parameter $\eta'$). We restrict our analysis to Beta distributions that are independent across arms. The details are in Section~\ref{sec:bayes}.

\subsection{Preliminaries}
\label{prelims-prelims}

\vspace{-2mm}\xhdr{Reward-tape.} It is convenient for our analyses to interpret the realized rewards of each arm as if they are written out in advance on a ``tape". We posit a matrix
    $\rbr{\tape_{a,i}\in [0,1]:\; a\in[2],\;i\in[T]}$,
called \emph{reward-tape}, such that each entry $\tape_{a,i}$ is an independent Bernoulli draw with mean $\mu_a$. This entry is returned as reward when and if arm $a$ is chosen for the $i$-th time. (We start counting from the initial samples, which comprise entries $i\in[N_0]$.) This is an equivalent (and well-known) representation of rewards in stochastic bandits.

We will use the notation for the UCBs/LCBs defined by the reward-tape. Fix arm $a\in[2]$ and $n\in [T]$. Let
    $\muTape_{a,n} = \tfrac{1}{n}\sum_{i\in[n]} \tape_{a,i}$
be the average over the first $n$ entries for arm $a$. Now, given $\eta\geq 0$, define the appropriate confidence bounds:
\begin{align}  \label{eq:def_ucb_tape}
\ucbTape_{a, n} := \min\cbr{1, \muTape_{a,n} + \sqrt{\eta/n}}
\quad\text{and}\quad
\lcbTape_{a, n} := \max\cbr{0, \muTape_{a,n} - \sqrt{\eta/n}}.
\end{align}

\xhdr{Good/bad arm.} When $\mu_1,\mu_2$ are fixed (rather than drawn from a prior), we posit that $\mu_1>\mu_2$. That is, arm $1$ is the \emph{good arm}, and arm $2$ is the \emph{bad arm}. (It is for the ease of notation, and not known to the algorithm.) Our guarantees depend on $\Delta := \mu_1-\mu_2$, called the \emph{gap} (between the two arms), a very standard quantity in multi-armed bandits.

\xhdr{The big-O notation.} We use the big-O notation to hide constant factors. Specifically, $O(X)$ and $\Omega(X)$ mean, resp., ``at most $c_0\cdot X$" and ``at least $c_0\cdot X$" for some absolute constant $c_0>0$ that is not specified in the paper. When and if $c_0$ depends on some other absolute constant $c$ that we specify explicitly, we point this out in words and/or by writing, resp., $O_c(X)$ and $\Omega_c(X)$.
%As usual, $\Theta(X)$ is a shorthand for ``both $O(X)$ and $\Omega(X)$", and writing $\Theta_c(X)$ emphasizes the dependence on $c$.

\xhdr{Bandit algorithms.} Algorithms UCB1 and Thompson Sampling achieve regret
\begin{align}
\regret(T) \leq O\rbr{\min\rbr{1/\Delta,\sqrt{T}}\cdot\log T}.
\end{align}
This regret rate is essentially optimal among all bandit algorithms: it is optimal up to constant factors for fixed $\Delta>0$, and up to $O(\log T)$ factors for fixed $T$ (see Section~\ref{sec:related} for citations).

A key property of a reasonable bandit algorithm is that $\regret(T)/T \to 0$; this property is also called \emph{no-regret}. Conversely, algorithms with $\regret(T) \geq \Omega(T)$ are considered very inefficient.

A bandit algorithm implemented by a collective of $\eta$-confident agents will be called an \emph{$\eta$-confident algorithm}. Likewise, \emph{$\eta$-optimistic algorithm} and \emph{$\eta$-pessimistic algorithm.}

%% file: freq_LB.tex
% !TEX root =  main.tex
\section{Learning failures}\label{sec:freq_homo}
\label{sec:fail}

In this section, we prove that the agents' myopic behavior causes learning failures, \ie all but a few agents choose the bad arm. More precisely:

\begin{definition}
The \emph{$n$-sampling failure} is an event that all but $\leq n$ agents choose the bad arm.
\end{definition}

Our main result allows arbitrary $\eta$-confident agents. Essentially, it asserts that $0$-sampling failures happen with probability at least $\pfail \sim e^{-O(\eta)}$. This is a stark learning failure when $\eta$ is a constant relative to the time horizon $T$.

We make two technical assumptions:
\begin{align}
&\text{arms' mean rewards lie in $(c, 1-c)$, for some absolute constant $c \in (0, \nicefrac12)$}, \label{eq:assn-1}\\
&\text{the number of initial samples satisfies
    $N_0 \geq 64\,\eta/c^2 + 1/c$}. \label{eq:assn-2}
\end{align}

\noindent The meaning of \eqref{eq:assn-1} is that it rules out degenerate behaviors when mean rewards are close to the known upper/lower bounds. The big-O notation hides the dependence on the absolute constant $c$, when and if explicitly stated so. Assumption \eqref{eq:assn-2} ensures that the $\eta$-confidence interval is a proper subset of $[0,1]$ for all agents; we sidestep this assumption later in Theorem~\ref{thm:LB-ext}.

Thus, the result is stated as follows:

\begin{theorem}[$\eta$-confident agents]\label{thm:interval}
Suppose all agents are $\eta$-confident, for some fixed $\eta\geq 0$. Make assumptions \eqref{eq:assn-1} and \eqref{eq:assn-2}. Then the $0$-sampling failure occurs with probability at least
\begin{align}\label{eq:thm:interval-1}
 \pfail = \Omega_c\rbr{\sqrt{(1+\eta)/N_0}}
    \cdot e^{-O_c\rbr{\eta \;+\; N_0\Delta^2}},
\quad\text{where}\quad \Delta = \mu_1-\mu_2.
\end{align}
Consequently, $ \regret(T) \geq \Delta\cdot \pfail\cdot T$.
 \end{theorem}

\begin{discussion}
The agents in Theorem~\ref{thm:interval} can exhibit any behaviors, possibly different for different agents and different arms, as long as these behaviors are consistent with the $\eta$-confidence property. In particular, this result applies to deterministic behaviours such as optimism/pessimism, and also to randomized behaviors such as $\eta$-Thompson agents defined in Section~\ref{sec:prelims-cases}.

From the perspective of multi-armed bandits, Theorem~\ref{thm:interval} implies that $\eta$-confident bandit algorithms with constant $\eta$ cannot be no-regret, \ie cannot have regret sublinear in $T$.

Note that the guarantee in Theorem~\ref{thm:interval} deteriorates as the parameter $\eta$ increases, and becomes essentially vacuous when $\eta\sim \log(T)$. The latter makes sense, since this regime of $\eta$ is used in UCB1 algorithm and suffices for Thompson Sampling.
\end{discussion}

\begin{discussion}
Assumption \eqref{eq:assn-2} is innocuous from the social learning perspective: essentially, the agents hold initial beliefs grounded in data and these beliefs are not completely uninformed. From the bandit perspective, this assumption is less innocuous: while it seems unreasonable to discard the initial data, an algorithm can always choose to do so, possibly side-stepping the failure result. In any case, we remove this assumption in Theorem~\ref{thm:LB-ext} below.
\end{discussion}

\begin{remark}
A weaker version of \eqref{eq:assn-2}, namely $N_0\geq \eta$, is necessary to guarantee an $n$-sampling failure for any $\eta$-confident agents. Indeed, suppose all agents are $\eta$-optimistic for arm $1$ (the good arm), and $\eta$-pessimistic for arm $2$ (the bad arm). If $N_0<\eta$, then the index for arm $2$ is $0$ after the initial samples, whereas the index of arm $1$ is always positive. Then all agents choose arm $1$.
\end{remark}

\medskip

Next, we spell out two corollaries which help elucidate the main result.

\begin{corollary}\label{cor:LB-small-gap}
If the gap is sufficiently small, $\Delta < O\rbr{1/\sqrt{N_0}}$, then Theorem~\ref{thm:interval} holds with
\begin{align}\label{eq:thm:interval-2}
 \pfail = \Omega_c\rbr{\sqrt{(1+\eta)/N_0}}\cdot e^{-O_c(\eta)}.
 \end{align}
\end{corollary}

\begin{remark}
The assumption in Corollary~\ref{cor:LB-small-gap} is quite mild in light of the fact that when $\Delta > \Omega\rbr{\sqrt{\log(T)/N_0}}$, the initial samples suffice to determine the best arm with high probability.
\end{remark}

\begin{corollary}\label{cor:LB-unbiased}
If all agents are unbiased, then Theorem~\ref{thm:interval} holds with $\eta=0$ and
\begin{align}
 \pfail
    &= \Omega_c\rbr{1/\sqrt{N_0}}\cdot e^{-O_c\rbr{N_0\;\Delta^2}}
            \label{eq:thm:interval-3}\\
    &= \Omega_c\rbr{1/\sqrt{N_0}}
        \qquad\qquad \text{if $\Delta < O\rbr{1/\sqrt{N_0}}$}.\nonumber
\end{align}
In the latter case,
$ \regret(T) \geq \Omega_c\rbr{\Delta/\sqrt{N_0}} \cdot T$.
\end{corollary}

\begin{remark}\label{rem:trivial-failure}
A trivial failure result relies on the event $\cE$ that all $N_0$ initial samples of the good arm are realized as $0$. ($\cE$ implies a $0$-sampling failure as long as $\geq 1$ initial sample of the bad arm is realized to $1$.) This result is weak for $N_0\gg 1$ since $\Pr{\cE} = (1-\mu_1)^{N_0}$. In contrast, our guarantee on the failure probability scales as $1/\sqrt{N_0}$ when the gap is small enough.
Thus, we have the first failure result for the greedy algorithm with a non-trivial dependence on $N_0$.
%$\pfail$ only depends on $N_0$ through the assumption that $\Delta < O\rbr{1/\sqrt{N_0}}$.
\end{remark}

%\begin{discussion}
%Corollary~\ref{cor:LB-unbiased} can be seen as a general result on the failure of the greedy algorithm. This is the first such result with a non-trivial dependence on $N_0$, to the best of our knowledge.
%\end{discussion}

\medskip

Let us remove assumption \eqref{eq:assn-2} and allow ``small" $N_0$, namely
    $N_0\leq N_* := \ceil{64\eta/c^2 + 1/c}$.
While the analysis of initial samples simplifies --- we rely on all samples being $0$ for the good arm and $1$ for the bad arm --- the rest of the analysis becomes more intricate. Essentially, this is due to ``boundary effects": confidence intervals are initially too wide to fit into the $[0,1]$ interval. The guarantee is slightly weaker: $n$-sampling failures, $n= N^*-N_0$, rather than $0$-sampling failures. Also, we need the behavioral type for each agent $t$ to satisfy two natural (and very mild) properties:
\begin{itemize}
\item[(P1)] \emph{(symmetry)} if all rewards in $\hist_t$ are $0$, the two arms are treated symmetrically;%
    \footnote{That is, the behavioral type stays the same if the arms' labels are switched.}

\item[(P2)] \emph{(monotonicity)} Fix any arm $a\in[2]$, any $t$-round history $H$ in which all rewards are $0$ for both arms, and any other $t$-round history $H'$ that contains the same number of samples of arm $a$ such that all these samples have reward $1$. Then
    \begin{align}\label{eq:thm:LB-ext-P2}
        \Pr{a_t=a \mid \hist_t = H'} \geq \Pr{a_t=a \mid \hist_t = H}.
    \end{align}
\end{itemize}
Note that both properties would still be natural and mild even without the ``all rewards are zero" clause. The resulting guarantee on the failure probability is somewhat cleaner.

\begin{theorem}[small $N_0$]\label{thm:LB-ext}
Fix $\eta\geq 0$, assume~\refeq{eq:assn-1}, and let $N_0\in [1, N^*]$, where
    $N^* := \ceil{64\eta/c^2 + 1/c}$.
Suppose each agent $t$ is $\eta$-confident and satisfies properties (P1) and (P2). Then an $n$-sampling failure,
    $n=N^*-N_0$,
occurs with probability at least
\begin{align}\label{eq:thm:LB-ext}
      \pfail = \Omega_c\rbr{ c^{2N^*}}
             = \Omega_c\rbr{e^{-O_c(\eta)}}.
\end{align}
Consequently,
    $ \regret(T) \geq \Delta\cdot \pfail\cdot (T-n)$.
\end{theorem}

\medskip
If all agents are pessimistic, we find that \emph{any levels of pessimism}, whether small or large or different across agents, lead to a $0$-sampling failure with probability $\Omega_c(1/\sqrt{N_0})$, matching Corollary~\ref{cor:LB-unbiased} for the unbiased behavior. This happens in the (very reasonable) regime when
\begin{align}\label{eq:thm:pessimism-LB-regime}
\Omega_c(\eta) < N_0 < O(1/\Delta^2).
\end{align}

\begin{theorem}[pessimistic agents]\label{thm:pessimis}
Suppose each agent $t\in[T]$ is $\eta_t$-pessimistic, for some $\eta_t\geq 0$. Suppose assumptions \eqref{eq:assn-1} and \eqref{eq:assn-2} hold for $\eta = \max_{t\in[T]} \eta_t$. Then the $0$-sampling failure occurs with probability lower-bounded by
\refeq{eq:thm:interval-3}. Consequently,
    $ \regret(T) \geq \Omega_c\rbr{\Delta/\sqrt{N_0}}\cdot e^{-O_c\rbr{N_0\,\Delta^2}}$.
\end{theorem}

Note that we allow extremely pessimistic agents ($\eta_t\sim \log T$), and that the pessimism level $\eta_t$ can be different for different agents $t$. The relevant parameter is $\eta = \max_{t\in[T]} \eta_t$, the highest level of pessimism among the agents. However, the failure probability in \eqref{eq:thm:interval-3} does not contain the $e^{-\eta}$ term.
%In particular, we obtain $\pfail = \Omega_c(\asedit{1/\sqrt{N_0}})$ when $N_0 < O(1/\Delta^2)$.
(The dependence on $\eta$ ``creeps in" through assumption \eqref{eq:assn-2},
\ie that $N_0>\Omega_c(\eta)$.)

\subsection{Proofs overview  and probability tools}
\label{sec:fail-overview}

Our proofs rely on two tools from Probability (proved in Appendix~\ref{app:tools}): a sharp  anti-concentration inequality for Binomial distribution and a lemma that encapsulates a martingale argument.

\begin{lemma}[anti-concentration]\label{lm:good_arm_sad}
Let $(X_i)_{i\in \N}$ be a sequence of independent Bernoulli random variables with mean $p\in [c,1-c]$, for some $c\in (0,\nicefrac12)$ interpreted as an absolute constant. Then
%for any $n\ge 1/c$ and $q > c/8$,
\begin{align}\label{eq:lm:good_arm_sad}
\rbr{\forall n\geq 1/c,\; q \in (c/8,\,p)}\qquad
\Pr{ \tfrac{1}{n}\;{\textstyle \sum_{i=1}^{n}}\; X_i \leq q}
    \ge \Omega_c\rbr{ e^{-O_c\rbr{n(p - q)^2}} }.
\end{align}
%where $\Omega_c(\cdot)$ and $O_c(\cdot)$ hide the dependence on $c$.
\end{lemma}

\begin{lemma}[martingale argument]\label{lm:bad_arm_happy}
In the setting of Lemma~\ref{lm:good_arm_sad},
\begin{align}\label{eq:lm:bad_arm_happy}
\forall q \in [0,p)\qquad
\Pr{\forall n\geq 1:\quad \tfrac{1}{n}\;{\textstyle \sum_{i=1}^{n}}\; X_i \geq q}
    \ge \Omega_c(p - q).
\end{align}
\end{lemma}

The overall argument will be as follows. We will use Lemma~\ref{lm:good_arm_sad} to upper-bound the average reward of arm 1, \ie the good arm, by some threshold $q_1$. This upper bound will only be guaranteed to hold when this arm is sampled exactly $N$ times, for a particular $N\geq N_0$.  Lemma~\ref{lm:bad_arm_happy} will allow us to uniformly \emph{lower}-bound the average reward of arm 2, \ie the bad arm, by some threshold $q_2\in(q_1, \mu_2)$. Focus on the round $t^*$ when the good arm is sampled for the $N$-th time (if this ever happens). If the events in both lemmas hold, from round $t^*$ onwards the bad arm will have a larger average reward by a constant margin $q_2-q_1$. We will prove that this implies that the bad arm has a larger index, and therefore gets chosen by the agents. The details of this argument differ from one theorem to another.

Lemma~\ref{lm:good_arm_sad} is a somewhat non-standard statement which follows from the anti-concentration inequality in \cite{zhang2020non} and a reverse Pinsker inequality in \cite{gotze2019higher}. More standard anti-concentration results via Stirling's approximation lead to an additional factor of $1/\sqrt{n}$ on the right-hand side of \eqref{eq:lm:good_arm_sad}. For Lemma~\ref{lm:bad_arm_happy}, we introduce an exponential martingale and relate the  event in \eqref{eq:lm:bad_arm_happy} to a deviation of this martingale. We then use Ville's inequality (a version of Doob's martingale inequality) to bound the probability that this deviation occurs.

\subsection{Proof of Theorem~\ref{thm:interval}: $\eta$-confident agents}
\label{sec:fail-mainpf}

%\begin{proof}[Proof Sketch of Theorem~\ref{thm:interval}]
%In order to prove the sampling failure,
%we will identify thresholds $q_1, q_2$ where $q_1 < q_2$
%and consider the following two events.
%\begin{enumerate}
%    \item The mean reward of arm 1 will after $N_0$ pulls is below $q_1$.
%    \item The mean reward of arm 2 is never below $q_2$ and
%\end{enumerate}
%We will show that if both of these events hold,
%arm 1 will not be played after the initial sampling phase.
%We prove this by bounding the difference
%between the index of each arm and their mean reward. Given the lower bound on $N_0$,
%we can upper bound this difference and show that it is smaller than the gap $q_2 - q_1$.
%Given the $q_2 - q_1$ gap between the mean rewards, this means that
%arm 1 will not be played after the initial sampling phase.
%
%To analyze the probability of the events, we analyse each separately and use
%the independence of the two reward sequence for the two arms to derive a lower bound on the joint probability.
%For the first event, we  use a sharp anti-concentration inequality for Bernoulli trials,
%i.e. Binomial distribution~\cite{zhang2020non},
%which lower bounds this probability based on the KL divergence of $\mu_1$ and $q_1$.
%We then apply the reserve Pinsker's inequality~\cite{gotze2019higher} to bound the KL divergence of two Bernoulli random variables.
%For the second event, we introduce an exponential
%martingale and relate the event
%to a deviation of the martingale from some threshold.
%We then use Doob's martingale inequality
%to bound the probability that this deviation occurs.
%\end{proof}

Fix thresholds $q_1<q_2$ to be specified later. Define two ``failure events":
\begin{description}
\item[$\sampfail_1$:] the average reward of arm 1 after the $N_0$ initial samples  is below $q_1$;
\item[$\sampfail_2$:] the average reward of arm 2 is never below $q_2$.
\end{description}
In a formula, using the reward-tape notation from Section~\ref{prelims-prelims}, these events are
\begin{align}
\sampfail_1 := \cbr{\muTape_{1,\,N_0} \le q_1}
\quad\text{and}\quad
\sampfail_2 := \cbr{\forall n \in [T]: \muTape_{2,n} \ge q_2}.
\label{eq:def-q-LB}
\end{align}

We show that event
    $\sampfail := \sampfail_1 \cap \sampfail_2$
implies the $0$-sampling failure, as long as the margin $q_2-q_1$ is sufficiently large.

\begin{claim}\label{claim:sampfail_implies_sample_failure}
Assume
    $q_2-q_1 > 2\cdot\sqrt{\eta/N_0}$
and event $\sampfail$. Then arm 1 is never chosen by the agents.
\end{claim}

\begin{proof}%[Proof of Claim~\ref{claim:sampfail_implies_sample_failure}]
Assume, for the sake of contradiction, that some agent chooses arm $1$. Let $t$ be the first round when this happens. Note that
    $\indx_{1, t} \ge \indx_{2, t}$.
We will show that this is not possible by upper-bounding $\indx_{1, t}$ and lower-bounding $\indx_{2, t}$.

By definition of round $t$, arm 1 has been previously sampled exactly $N_{0}$ times. Therefore,
\begin{align*}
\indx_{1, t}
    &\leq \muTape_{1,\,N_0} +  \sqrt{\eta/N_0}
        &\EqComment{by definition of index} \\
    &\leq q_1 + \sqrt{\eta/N_0}
        &\EqComment{by $\sampfail_1$} \\
    &< q_2 - \sqrt{\eta/N_0}
        &\EqComment{by assumption}.
    \end{align*}

Let $n$ be the number of times arm 2 has been sampled before round $t$. This includes the initial samples, so $n\ge N_0$. It follows that
\begin{align*}
\indx_{2, t}
    &\geq \muTape_{2,n} -  \sqrt{\eta/n}
        &\EqComment{by definition of index} \\
    &\geq q_2 - \sqrt{\eta/N_0}
        &\EqComment{by $\sampfail_2$ and $n\geq N_0$}.
\end{align*}
Consequently, $\indx_{2, t} > \indx_{1, t}$, contradiction.
\end{proof}

In what follows, let $c$ be the absolute constant from assumption \eqref{eq:assn-1}.

Let us lower bound $\Pr{\sampfail}$ by applying Lemmas~\ref{lm:good_arm_sad} and~\ref{lm:bad_arm_happy} to the reward-tape.

\begin{claim}\label{cl:pf:thm:interval-prob}
Assume $c/4<q_1<q_2<\mu_2$. Then
\begin{align}\label{eq:pf:thm:interval-prob}
\Pr{\sampfail}
    \geq \qfail := \Omega_c(\mu_2 - q_2) \cdot e^{-O_c\rbr{N_0(\mu_1 - q_1)^2}}.
\end{align}
\end{claim}
\begin{proof}
To handle $\sampfail_1$, apply Lemma~\ref{lm:good_arm_sad} to the reward-tape for arm $1$, \ie to the random sequence $(\tape_{1,i})_{i\in[T]}$, with $n=N_0$ and $q=q_1$. Recalling that $N_0 \geq 1/c$ by assumption \eqref{eq:assn-2},
\begin{align}\label{eq:pf:thm:interval-prob-1}
\Pr{\sampfail_1} \geq
    \Omega_c\rbr{e^{-O_c\rbr{N_0(\mu_1 - q_1)^2}}}.
\end{align}

To handle $\sampfail_2$, apply Lemma~\ref{lm:bad_arm_happy}  to the reward-tape for arm $2$, \ie to the random sequence $(\tape_{2,i})_{i\in[T]}$, with threshold $q=q_2$. Then
\begin{align}\label{eq:pf:thm:interval-prob-2}
\Pr{\sampfail_2} \geq \Omega_c(\mu_2 - q_2).
\end{align}
Events $\sampfail_1$ and $\sampfail_2$ are independent, because they are determined by, resp., realized rewards of arm $1$ and realized rewards of arm $2$. The claim follows.
\end{proof}
Finally, let us specify suitable thresholds that satisfy the preconditions in
Claims~\ref{claim:sampfail_implies_sample_failure}
and~\ref{cl:pf:thm:interval-prob}:
\begin{align*}
q_1 := \mu_2 - 4\cdot\sqrt{\eta/N_0} -
    %c\Delta/4
    c'/\sqrt{N_0}
    \quad\text{and}\quad
q_2 := \mu_2 - \sqrt{\eta/N_0} -
    %c\Delta/4
    c'/\sqrt{N_0},
\end{align*}
% \ssedit{where $c' = \sqrt{c}/4$. Plugging in $\mu_2\geq c$, $N_0\geq 64\cdot\eta/c^2$ for the second term, and $N_0 \ge 1/c$ for the third term, it is easy to check that $q_1\geq c/4$, as needed for Claim~\ref{cl:pf:thm:interval-prob}.}
where $c' = c/4$.
Plugging in $\mu_2\geq c$ and $N_0\geq \max(64\cdot\eta/c^2,1)$, it is easy to check that $q_1\geq c/4$, as needed for Claim~\ref{cl:pf:thm:interval-prob}.
Thus, the preconditions in
Claims~\ref{claim:sampfail_implies_sample_failure}
and~\ref{cl:pf:thm:interval-prob} are satisfied. It follows that the $0$-failure happens with probability at least $\qfail$, as defined in Claim~\ref{cl:pf:thm:interval-prob}. We obtain the final expression in \refeq{eq:thm:interval-1} because
    $\mu_1-q_1\leq O_c(\Delta + \sqrt{(1+\eta)/N_0})$
and
    $\mu_2-q_2\geq \Omega_c(\sqrt{(1+\eta)/N_0})$.

\subsection{Proof of Theorem~\ref{thm:pessimis}: pessimistic agents}
\label{sec:fail-pessimistic}

We reuse the machinery from Section~\ref{sec:fail-mainpf}: we define event
    $\sampfail := \sampfail_1 \cap \sampfail_2$
as per \refeq{eq:def-q-LB}, for some thresholds $q_1<q_2$ to be specified later, and use Claim~\ref{cl:pf:thm:interval-prob} to bound $\Pr{\sampfail}$. However, we need a different argument to prove that $\sampfail$ implies the $0$-sampling failure, and a different way to set the thresholds.

\begin{claim}\label{claim:pessimism_sampfail_implies_sample_failure}
Assume $q_1 > \sqrt{\eta/N_0}$ and event $\sampfail$. Then arm 1 is never chosen by the agents.
\end{claim}

\begin{proof}%[Proof of Claim~\ref{claim:sampfail_implies_sample_failure}]
Assume, for the sake of contradiction, that some agent chooses arm $1$. Let $t$ be the first round when this happens. Note that
    $\indx_{1, t} \ge \indx_{2, t}$.
We will show that this is not possible by upper-bounding $\indx_{1, t}$ and lower-bounding $\indx_{2, t}$.

By definition of round $t$, arm 1 has been previously sampled exactly $N_{0}$ times. Therefore,
\begin{align*}
\indx_{1, t}
    &= \max\{0, \muTape_{1,\,N_0} -\sqrt{\eta/N_0}\}
        &\EqComment{by definition of index} \\
    &\leq \max\{0, q_1 - \sqrt{\eta/N_0}\}
        &\EqComment{by $\sampfail_1$} \\
    &= q_1 - \sqrt{\eta/N_0}
    % &< q_2 - \sqrt{\eta/N_0}
        &\EqComment{by assumption}.
\end{align*}

Let $n$ be the number of times arm 2 has been sampled before round $t$. This includes the initial samples, so $n\ge N_0$. It follows that
\begin{align*}
\indx_{2, t}
    &\geq \muTape_{2,n} -  \sqrt{\eta/n}
        &\EqComment{by definition of index} \\
    &\geq q_2 - \sqrt{\eta/N_0}
        &\EqComment{by $\sampfail_2$ and $n\geq N_0$}.
%    &\geq q_1 - \sqrt{\eta/N_0}
%        &\EqComment{by $q_2 > q_1$}.
\end{align*}
Since $q_2>q_1$, it follows that $\indx_{2, t} > \indx_{1, t}$, contradiction.
\end{proof}

Now, set the thresholds $q_1,q_2$ as follows:
\begin{align*}
    q_1 := \mu_2 -
        %c\Delta/4
        2c'/\sqrt{N_0}
        \quad\text{and}\quad
    q_2 := \mu_2 -
        %c\Delta/8
        c'/\sqrt{N_0},
    \end{align*}
% \ssedit{where $c' = \sqrt{c}/8$.
% Plugging in $\mu_2 \ge c$ and $N_0 \ge 1/c$, we can check $q_1 > c/4$, thereby satisfying the precondition of Claim~\ref{cl:pf:thm:interval-prob}.
% Moreover, since we have $N_0 \ge 64\eta/c^2$, it is easy to see that $q_1 > c/4 > \sqrt{\eta/N_0}$.
% }
where $c' = c/8$.
Plugging in $\mu_2\geq c$ and $N_0\geq \max(64\cdot\eta/c^2,1)$,
it is easy to check that
$q_1 > \sqrt{\eta/N_0}$ and $q_1 \ge c/4$
as needed for Claim~\ref{cl:pf:thm:interval-prob} and
Claim \ref{claim:pessimism_sampfail_implies_sample_failure} respectively.
Thus, the preconditions in
Claims~\ref{cl:pf:thm:interval-prob}
and~\ref{claim:pessimism_sampfail_implies_sample_failure} are satisfied. So, the $0$-failure happens with probability at least $\qfail$ from  Claim~\ref{cl:pf:thm:interval-prob}. The final expression in \refeq{eq:thm:interval-1} follows because
    $\mu_1-q_1 \leq O_c(\Delta + 1/\sqrt{N_0})$
and
    $\mu_2-q_2 = \Omega_c(1/\sqrt{N_0})$.

\medskip
%\begin{proposition}\label{prop:const_regret}
%  Consider two Bernoulli arms with parameter $\mu_1>\mu_2$.
%  Let $\Delta = \mu_1 - \mu_2$.
%  Suppose that the agents suffer an $n$-sampling with probability $\pfail{}$.
%  $\pfail{}\Delta(T-n)$.
%\end{proposition}

%\subsection{Proof of Proposition \ref{prop:const_regret}}
%\begin{proof}
%  Let \sampfail{} denote the event that arm 1 is not played more than $N$ times.
%  By assumption $\Pr{\sampfail} \ge P$.
%  We can therefore lower bound the regret as
%  \begin{align*}
%    \Ex{\reg_T} &= \Delta\Ex{n_{T}(2)}
%    \\&=
%    \Delta \Ex{n_{T}(2) | \sampfail}\Pr{\sampfail} + \Ex{n_T(2) | \sampfail^{c}}\Pr{\sampfail^{c}}
%    \\&\ge
%    \Delta \Ex{n_{T}(2) | \sampfail}\Pr{\sampfail}
%    % &\text{Since $n_T(2) \ge 0$}
%    \\&\ge
%    \Delta P(T - N)
%    % &\text{Definition of sampling failure}
%  \end{align*}
%\end{proof}

%% file: app-LB-extension.tex
\subsection{Proof of Theorem~\ref{thm:LB-ext}: small $N_0$}
\label{sec:fail-extension}

We focus on the case when
    $N_0\leq N^* := \ceil{64\eta/c^2 + 1/c}$.
We can now afford to handle the initial samples in a very crude way: our failure events posit that all initial samples of the good arm return reward $0$, and all initial samples of the bad arm return reward $1$.
\begin{align*}
\sampfail_1
    &:= \setnot{\forall i \in [1, N^*]: \tape_{1, i} = 0},\\
\sampfail_2
    &:= \setnot{
      \forall i \in [1, N^*]: \tape_{2, i} = 1 \quad\text{ and }\quad
      \forall i\in[T]: \muTape_{2, i} \ge q_2}.
\end{align*}
Here, $q_2>0$ is the threshold to be defined later.

On the other hand, our analysis given these events becomes more subtle. In particular, we introduce another ``failure event" $\sampfail_3$, with a more subtle definition: if arm 1 is chosen by at least
    $n:=N^* - N_0$
agents, then arm 2 is chosen by $n$ agents before arm 1 is.

%Equivalently, let $\tau_n$ be the $n$-th agent when arm $1$

We first show that $\sampfail := \sampfail_1 \cap \sampfail_2 \cap \sampfail_3$ implies the $n$-sampling failure.

% \kbcomment{At some point we need to take a ceil before we talk about integers. We can take
% a ceil when we define $N^*$ in the theorem. We can also do $n = \ceil{N^*} - N_0$.
% The latter is a quicker fix but the former is in my opinion cleaner.
% }
\begin{claim}\label{cl:sampfail_implies_samp_fail_hetero}
Assume that $q_2 \ge c/4$ and $\sampfail$ holds. Then at most $n = N^* -N_0$ agents choose arm 1.
\end{claim}

\begin{proof}
  For the sake of contradiction, suppose arm 1 is chosen by more than $n$ agents.
  Let agent $t$ be the $(n+1)$-th agent that chooses arm $1$. In particular, $\indx_{1, t} \geq \indx_{2, t}$.

By definition of $t$, arm 1 has been previously sampled exactly $N^*$ times before (counting the $N_0$ initial samples). Therefore,
\begin{align*}
\indx_{1, t}
    &\leq \muTape_{1, N^*} +  \sqrt{\eta/N^*}
        &\EqComment{by $\eta$-confidence}\\
    &=\sqrt{\eta/N^*}
        &\EqComment{by event $\sampfail_1$}\\
    &\leq c/8
        &\EqComment{by definition of $N^*$}.
\end{align*}
Let $m$ be the number of times arm 2 has been sampled before round $t$. Then
\begin{align*}
\indx_{2, t}
  &\geq \muTape_{2, m} - \sqrt{\eta/m}
    &\EqComment{by $\eta$-confidence}\\
  &\geq q_2 - \sqrt{\eta/m}
    &\EqComment{by event $\sampfail_2$}\\
  &\geq q_2 - \sqrt{\eta/N^*}
    &\EqComment{since $m\ge N^*$ by event $\sampfail_3$}\\
  &\geq q_2 - c/8
    &\EqComment{by definition of $N^*$}\\
  &> c/8
    &\EqComment{since $q_2\geq c/4$}.
\end{align*}
Therefore, $\indx_{2, t} > \indx_{1, t}$, contradiction.
\end{proof}

Next, we lower bound the probability of $\sampfail_1 \cap \sampfail_2$ using Lemma \ref{lm:bad_arm_happy}.
\begin{claim}\label{cl:pr_samp12_hetero}
If $q_2<\mu_2$ then
  $\Pr{\sampfail_1 \cap \sampfail_2}
    \geq \Omega_c(\mu_2 - q_2)\cdot c^{2\,N^*}$.
\end{claim}

\begin{proof}
Instead of analyzing $\sampfail_2$ directly, consider events
\begin{align*}
\cE :=\cbr{ \forall i \in [1, N^*]: \tape_{2, i} = 1 }
\text{ and }
\cE' :=\cbr{\forall m \in [N^* + 1, T]:
                \tfrac{1}{m-N^*}\;{\textstyle \sum_{i=N^* + 1}^{m}}\;\tape_{2, i} \ge q_2}.
\end{align*}
Note that $\cE\cap\cE'$ implies $\sampfail_2$. Now,
    $\Pr{\sampfail_1}\geq {\mu_1}^{N^*} \geq c^{N^*}$
and
    $\Pr{\cE} \geq (1-\mu_2)^{N^*} \geq c^{N^*}$.
Further,
    $\Pr{\cE'}\geq \Omega_c(\mu_2-q_2)$
by Lemma \ref{lm:bad_arm_happy}. The claim follows since these three events are mutually independent.
\end{proof}

To bound $\Pr{\sampfail}$, we argue indirectly, assuming $\sampfail_1 \cap \sampfail_2$ and proving that the conditional probability of $\sampfail_3$ is at least $\nicefrac12$. While this statement feels natural given that $\sampfail_1 \cap \sampfail_2$ favors arm $2$, the proof requires a somewhat subtle inductive argument. This is where we use the symmetry and monotonicity properties from the theorem statement.

\begin{claim}\label{cl:pr_samp3}
  $\Pr{\sampfail_3 \mid \sampfail_1 \cap \sampfail_2} \ge \frac{1}{2}$.
\end{claim}

Now, we can lower-bound $\Pr{\sampfail}$ by
    $\Omega_c(\mu_2 - q_2)\cdot c^{2\,N^*}$.
Finally, we set the threshold to $q_2=c/2$ and the theorem follows.

\begin{proof}[Proof of Claim~\ref{cl:pr_samp3}]
Note that event $\sampfail_t$ is determined by the first $N^*$ entries of the reward-tape for both arms, in the sense that it does not depend on the rest of the reward-tape.

For each arm $a$ and $i\in[T]$, let agent $\tau_{a,i}$ be the $i$-th agent that chooses arm $a$, if such agent exists, and $\tau_i=T+1$ otherwise. Then
\begin{align}\label{eq:cl:pr_samp3-equiv}
\sampfail_3
    = \cbr{\tau_{2,n}\leq \tau_{1,n}}
    = \cbr{\tau_{1,n} \geq 2n }
\end{align}

Let $\cE$ be the event that the first $N^*$ entries of the reward-tape are $0$ for both arms. By symmetry between the two arms (property (P1) in the theorem statement) we have
\[ \Pr{\tau_{2,n}<\tau_{1,n}\mid \cE} = \Pr{\tau_{2,n}>\tau_{1,n}\mid \cE} = \nicefrac12,\]
and therefore
\begin{align}\label{eq:cl:pr_samp3-half}
\Pr{\sampfail_3\mid \cE}
    = \Pr{\tau_{2,n}\leq \tau_{1,n}\mid \cE}
    \geq \nicefrac12.
\end{align}

Next, for two distributions $F,G$, write $F\geqfosd G$ if $F$ first-order stochastically dominates $G$. A conditional distribution of random variable $X$ given event $\cE$ is denoted $(X|\cE)$. For each $i\in[T]$, we consider two conditional distributions for $\tau_{1,i}$: one given $\sampfail_1\cap\sampfail_2$ and another given $\cE$, and prove that the former dominates:
\begin{align}\label{eq:cl:pr_samp3-fosd}
\rbr{\tau_{1,i} \mid \sampfail_1\cap\sampfail_2} \geqfosd
\rbr{\tau_{1,i} \mid \cE}
\quad\forall i\in[T].
\end{align}

\noindent Applying \eqref{eq:cl:pr_samp3-fosd} with $i=n$, it follows that
\begin{align*}
\Pr{\sampfail_3\mid \sampfail_1\cap\sampfail_2}
    &= \Pr{\tau_{1,n} \geq 2n \mid \sampfail_1\cap\sampfail_2} \\
    &\geq \Pr{\tau_{1,n} \geq 2n \mid \cE}
    = \nicefrac12.
\end{align*}
(The last equality follows from \eqref{eq:cl:pr_samp3-half} and \refeq{eq:cl:pr_samp3-half}.)
Thus, it remains to prove \eqref{eq:cl:pr_samp3-fosd}.

Let us consider a fixed realization of each agents' behavioral type, \ie a fixed, deterministic mapping from histories to arms. W.l.o.g. interpret the behavioral type of each agent $t$ as first deterministically mapping history $\hist_t$ to a number $p_t\in[0,1]$, then drawing a threshold $\theta_t\in[0,1]$ independently and uniformly at random, and then choosing arm $1$ if and only if $p_t\geq \theta_t$. Note that
    $p_t = \Pr{a_t=1\mid \hist_t}$.
So, we pre-select the thresholds $\theta_t$ for each agent $t$. Note the agents retain the monotonicity property (P2) from the theorem statement. (For this property, the probabilities on both sides of \refeq{eq:thm:LB-ext-P2} are now either $0$ or $1$.)

Let us prove \eqref{eq:cl:pr_samp3-fosd} for this fixed realization of the types, using induction on $i$. Both sides of \eqref{eq:cl:pr_samp3-fosd} are now deterministic; let $A_i,B_i$ denote, resp., the left-hand side and the right-hand side. So, we need to prove that $A_i\geq B_i$ for all $i\in [n]$. For the base case, take $i=0$ and define $A_0 = B_0 = 0$. For the inductive step, assume $A_i\geq B_i$ for some $i\geq 0$. We'd like to prove that $A_{i+1}\geq B_{i+1}$. Suppose, for the sake of contradiction, that this is not the case, \ie $A_{i+1}< B_{i+1}$. Since $A_i<A_{i+1}$ by definition of the sequence $(\tau_{a,i}:\,\in[T])$, we must have
    \[ B_i \leq A_i < A_{i+1} <B_{i+1}. \]
Focus on round $t = A_{i+1}$. Note that the history $\hist_t$ contains exactly $i$ agents that chose arm $1$, both under event $\sampfail_1\cap\sampfail_2$ and under event $\cE$. Yet, arm $2$ is chosen under $\cE$, while arm $1$ is chosen under $\sampfail_1\cap\sampfail_2$. This violates the monotonicity property (P2) from the theorem statement. Thus, we've proved \eqref{eq:cl:pr_samp3-fosd} for any fixed realization of the types. Consequently, \eqref{eq:cl:pr_samp3-fosd} holds in general.
\end{proof}

%% file: freq_UB.tex
\section{Upper bounds for optimistic agents}
\label{sec:UB}

In this section, we upper-bound regret for optimistic agents. We match the exponential-in-$\eta$ scaling from Corollary~\ref{cor:LB-small-gap}. Further, we refine this result to allow for different behavioral types.

On a technical level, we prove three regret bounds of the same shape \eqref{eq:thm:upper}, but with a different $\Phi$ term. (We adopt a unified presentation to emphasize this similarity.) Throughout, $\Delta = \mu_1-\mu_2$ denotes the gap between the two arms.

The basic result assumes that all agents have the same behavioral type.

\begin{theorem}\label{thm:upper}
Suppose all agents are $\eta$-optimistic, for some fixed $\eta>0$. Then, letting $\Phi=\eta$,
% \begin{align}\label{eq:thm:upper}
% \regret(T)
%     \leq O_c\rbr{ T\cdot e^{-\Omega_c(\eta)}\cdot \Delta\log\rbr{\nicefrac{1}{\Delta}}
%                     \;+\;\frac{\eta}{\Delta}}.
% \end{align}
\begin{align}\label{eq:thm:upper}
\regret(T)
    \leq O\rbr{ T\cdot e^{-\Omega(\eta)}\cdot
        \Delta(1 + \log(\nicefrac{1}{\Delta}))
    \;+\;\frac{\Phi}{\Delta}}.
\end{align}
\end{theorem}

\begin{discussion}
The main take-away is that the exponential-in-$\eta$ scaling from Corollary~\ref{cor:LB-small-gap} is tight for $\eta$-optimistic agents, and therefore the best possible lower bound that one could obtain for $\eta$-confident agents. This result holds for any given $N_0$, the number of initial samples.%
\footnote{For ease of exposition, we do not track the improvements in regret when $N_0$ becomes larger.}
Our guarantee remains optimal in the ``extreme optimism" regime when $\eta\sim\log(T)$, whereby it matches the optimal regret rate,
    $O\rbr{\tfrac{\log T}{\Delta}}$,
for large enough $\eta$.
\end{discussion}

What if different agents can hold different behavioral types? First, let us allow agents to have varying amounts of optimism, possibly different across arms and possibly randomized.

\begin{definition}
Fix $\etamax\geq \eta>0$. An agent $t\in [T]$ is called \emph{$\sbr{\eta,\etamax}$-optimistic} if its index $\indx_{a,t}$ lies in the interval
    $\sbr{\ucb_{a, t}, \ucb[\etamax]_{a, t}}$,
for each arm $a\in[2]$.
\end{definition}

We show that the guarantee in Theorem~\ref{thm:upper} is robust to varying the optimism level ``upwards".

\begin{theorem}[robustness]\label{thm:upper-interval}
Fix $\etamax\geq \eta>0$. Suppose all agents are $\sbr{\eta,\etamax}$-optimistic. Then
regret bound \eqref{eq:thm:upper} holds with $\Phi = \etamax$.
%\begin{align}\label{eq:thm:upper-interval}
%\regret(T)
%    \leq O\rbr{ T\cdot e^{-\Omega(\eta)}\cdot \Delta(1 + \log(\nicefrac{1}{\Delta}))
%                    \;+\;\frac{\etamax}{\Delta}}.
%\end{align}
\end{theorem}

Note that the upper bound $\etamax$ has only a mild influence on the regret bound in Theorem~\ref{thm:upper-interval}.

Our most general result only requires a small fraction of agents to be optimistic, whereas all agents are only required to be $\etamax$-confident (allowing all behaviors consistent with that).

\begin{theorem}[recurring optimism]\label{thm:upper-distr}
Fix $\etamax\geq \eta>0$. Suppose all agents are $\etamax$-confident. Further, suppose
each agent's behavioral type is chosen independently at random so that the agent is $\sbr{\eta,\etamax}$-optimistic with probability at least $q>0$.
Then regret bound \eqref{eq:thm:upper} holds with $\Phi = \etamax/q$.
\end{theorem}

\begin{discussion}
The take-away is that once there is even a small fraction of optimists,
    $q>\tfrac{1}{\Delta\cdot o(T)}$,
the behavioral type of less optimistic agents does not have much impact on regret. In particular, it does not hurt much if they become very pessimistic. A small fraction of optimists goes a long way!

Note that a small-but-constant fraction of \emph{extreme} optimists, \ie $\eta,\etamax\sim \log(T)$ in Theorem~\ref{thm:upper-distr}, yields optimal regret rate, $\log(T)/\Delta$.
\end{discussion}

\subsection{Proof of Theorem~\ref{thm:upper} and Theorem~\ref{thm:upper-interval}}
\label{sec:UB-proof-1}

We define certain ``clean events" to capture desirable realizations of random rewards, and  decompose our regret bounds based on whether or not these events hold. The ``clean events" ensure that the index of each arm is not too far from its true mean reward; more specifically, that the index is ``large enough" for the good arm, and ``small enough" for the bad arm. We have two ``clean events", one for each arm, defined in terms of the reward-table as follows:
\begin{align}
\clean^{\eta}_1
    &:= \cbr{ \forall i \in [T]:\; \ucbTape_{1, i} \geq \mu_1 - \Delta/2 },
        \label{eq:def_clean1}\\
\clean^{\eta}_{2}
    &:=\cbr{ \forall i \geq 64\,\eta/\Delta^2:\;  \ucbTape_{2, i} \leq \mu_2 + \Delta/4}.
        \label{eq:def_clean2}
\end{align}

Our analysis is more involved compared to the standard analysis of the UCB1 algorithm \cite{bandits-ucb1}, essentially because we cannot make $\eta$ be ``as large as needed" to ensure that clean events hold with very high probability. For example, we cannot upper-bound the deviation probability separately for each round and naively take a union bound over all rounds.%
\footnote{Indeed, this would only guarantee that clean events hold with probability at least
    $1-O(T\cdot e^{-\Omega(\eta)})$,
which in turn would lead to a regret bound like $O(T^2\cdot e^{-\Omega(\eta)})$.}
Instead, we apply a more careful ``peeling technique", used \eg in \citet{Bubeck-colt09}, so as to avoid \emph{any} dependence on $T$ in the lemma below.
% \sscomment{Q: Are we not going to mention that using peeling technique with max concentration, equipped with standard clean events (which does not consider $\Delta$ as above) only give us $T \ln T$ bound? (so we instead use a clever clean event defined as above) I think we partially discussed this one in the meeting, but I could not exactly remember.}

\begin{lemma}\label{lm:failure_total_simplify}
The clean events hold with probability
\begin{align}
\Pr{\clean^{\eta}_1}
    &\geq 1- O\rbr{(1 + \log(\nicefrac{1}{\Delta}))\cdot e^{-\Omega(\eta)}},
    \label{eq:lm:failure_total-1}\\
\Pr{\clean^{\eta}_2}
    &\geq 1- O\rbr{e^{-\Omega(\eta)}}.
    \label{eq:lm:failure_total-2}
\end{align}
\end{lemma}
%Note that events $\clean^{\eta}_1$ and $\clean^{\etamax}_2$ are independent, since they are determined by, resp., realized rewards of arm $1$ and realized rewards of arm $2$.

We show that under the appropriate clean events, $\eta$-optimistic agents cannot play the bad arm too often. In fact, this claim extends to $[\eta,\etamax]$-optimistic agents.

\begin{claim}\label{cl:clean_implies}
Assume that events $\clean^{\eta}_1$ and $\clean^{\etamax}_2$ hold. Then $[\eta,\etamax]$-optimistic agents cannot choose the bad arm more than $64\,\etamax/\Delta^2$ times.
\end{claim}

\begin{proof}
For the sake of contradiction, suppose $[\eta,\etamax]$-optimistic agents choose the bad arms at least
$n=64\,\etamax/\Delta^2$ times, and let $t$ be the round when this happens. However, by event $\clean^{\eta}_1$, the index of arm 1 is at least $\mu_1 - \Delta/2$. By event $\clean^{\etamax}_2$, the index of arm 2 is at most
% at least
    $\ucbTape_{i,n} \leq \mu_2 + \Delta/4$,
which is less than the index of arm 1, contradiction.
\end{proof}

For the ``joint" clean event,
    $\clean := \clean^{\eta}_1 \cap \clean^{\etamax}_2$,
Lemma~\ref{lm:failure_total_simplify} implies
\begin{align}\label{eq:UB-clean-prob}
\Pr{\clean}
    \geq 1- O\rbr{\log\rbr{\nicefrac{1}{\Delta}}\cdot e^{-\Omega(\eta)}}.
\end{align}

When the clean events fail, we upper-bound regret by $\Delta\cdot T$, which is the largest possible. Thus, Lemma~\ref{cl:clean_implies} and \refeq{eq:UB-clean-prob} imply Theorem~\ref{thm:upper-interval}, which in turn implies  Theorem~\ref{thm:upper} as a special case.

\subsection{Proof of Theorem~\ref{thm:upper-distr}}
\label{sec:UB-proof-2}

We reuse the machinery from Section~\ref{sec:UB-proof-1}, but we need some extra work. Recall that all agents are assumed to be $\etamax$-confident, whereas only a fraction are optimistic. Essentially, we rely on the optimistic agents to sample the good arm sufficiently many times (via Claim~\ref{cl:clean_implies}). Once this happens, all other agents ``fall in line" and cannot choose the bad arm too many times.

In what follows, let
    $m =1+64\,\etamax/\Delta^2$.

\begin{claim}\label{cl:clean_implies_more}
Assume $\clean$. Suppose the good arm is sampled at least $m$ times by some round $t_0$. Then after round $t_0$, agents cannot choose the bad arm more than $m$ times.
\end{claim}

\begin{proof}
  For the sake of contradiction, suppose agent $t\geq t_0$ has at least $m$ samples of the bad arm (\ie $n_{2,t}\geq m$), and chooses the bad arm once more. Then the index of the good arm satisfies
\begin{align*}
\indx_{1,t}
    &\geq \lcb[\etamax]_{1,t}
        &\EqComment{$\etamax$-confident agents} \\
    &\geq \lcbTape[\etamax]_{1,m}
        &\EqComment{by definition of $t_0$} \\
    &\geq \ucbTape[\etamax]_{1,m} - 2\,\sqrt{\etamax/m}
        &\EqComment{by definition of UCBs/LCBs} \\
    &\geq \ucbTape_{1,m} - 2\,\sqrt{\etamax/m}
        &\EqComment{since $\etamax\geq \eta$} \\
    &> \mu_1-\Delta/2
        &\EqComment{by $\clean^{\eta}_1$ and the definition of $m$}.
\end{align*}
The index of the bad arm satisfies
\begin{align*}
\indx_{2,t}
    &\leq \ucb_{1,t}
        &\EqComment{$\eta$-confident agents} \\
    &\leq \mu_2+\Delta/4
        &\EqComment{by $\clean^{\eta}_1$ and the definition of $m$},
\end{align*}
which is strictly smaller than $\indx_{1,t}$, contradiction.
\end{proof}

For Claim~\ref{cl:clean_implies_more} to ``kick in", we need sufficiently many optimistic agents to arrive by time $t_0$. Formally, let $\cE_t$ be the event that at least $2m$ agents are  $\sbr{\eta,\etamax}$-optimistic in the first $t$ rounds.

\begin{corollary}\label{cor:clean_implies_more}
Assume $\clean$. Further, assume event $\cE_{t_0}$ for some round $t_0$. Then (by Claim~\ref{cl:clean_implies}) the good arm is sampled at least $m$ times before round $t_0$. Consequently (by Claim~\ref{cl:clean_implies_more}), agents cannot choose the bad arm more than $m+t_0$ times.
\end{corollary}

Finally, it is easy to see by Chernoff Bounds that
    $\Pr{\cE_{t_0}} \geq 1-e^{-\Omega(\eta)}$
for some $t_0 = O(m/q)$, where $q$ is the probability from the theorem statement.
%Further, event $\cE_{t_0}$ is independent from the clean events.
So, $\Pr{\clean\cap\cE_{t_0}}$ is lower-bounded as in \refeq{eq:UB-clean-prob}. Again, when $\clean\cap\cE_{t_0}$ fails, we upper-bound regret by $\Delta\cdot T$. So, Corollary~\ref{cor:clean_implies_more} and the lower bound on $\Pr{\clean\cap\cE_{t_0}}$ implies the theorem.

%\begin{proof}[Proof of Theorem \ref{thm:upper}]
%  Let $\comp{\clean}$ denote the complement even of $\clean$.
%  We can bound the regret as
%  \begin{align*}
%    \Ex{\reg_T} &=
%    \Ex{ \Delta n_{T}(2)}
%    \\&=
%    \Ex{ \Delta n_{T}(2) \big| \clean}\Pr{\clean}
%    +
%    \Ex{ \Delta n_{T}(2) \big| \comp{\clean}}\Pr{\comp{\clean}}
%    \\&\le
%    O\left( \frac{64\eta}{\Delta} + T\Delta\log(\frac{1}{\Delta})e^{-\Omega(\eta)} \right),
%  \end{align*}
%  where in the inequality, we have bounded the first term using Claim \ref{claim:clean_implies_low_regert} and the second term using Lemma \ref{lm:failure_total_simplify}.
%\end{proof}

%% file: sec-Bayes.tex
\section{Learning failures for Bayesian agents}
\label{sec:bayes}

This section is on \emph{Bayesian agents}. That is, we posit that agents are endowed with Bayesian beliefs, form posteriors given the observed data, and act according to these posteriors. The Bayesian beliefs are same for all agents, and independent across arms unless specified otherwise.

\xhdr{Formal model.}
Agents believe that mean rewards $(\mu_1,\mu_2)$ are initially drawn from some distribution $\cP$ over $[0,1]^2$. Each agent $t$ computes a joint posterior $\cP_t$ on $(\mu_1,\mu_2)$ given the history $\hist_t$, and acts according to this posterior. (The history contains $N_0$ initial samples from each arm, as before.) $\cP$ and $\cP_t$ are also called \emph{beliefs}: resp., prior beliefs and (agent-$t$) posterior beliefs. Note that the Bayesian update for agent $t$ is determined by the history $\hist_t$, and does not depend on the beliefs of the previous agents.

We posit a fixed bandit instance $(\mu_1,\mu_2)$ throughout this section. Given the prior beliefs, the posteriors $\cP_t$ are well-defined, regardless of how $(\mu_1,\mu_2)$ is \emph{actually} chosen. (In Section~\ref{sec:priors}, we consider \emph{Bayesian bandits}, when the bandit instance is actually sampled from $\cP$.)

We assume \emph{independent beliefs}:
agents believe that each $\mu_a$, $a\in[2]$ is drawn independently from some distribution $\cP_{a,0}$ over [0,1], so that
    $\cP = \cP_{1,0}\times \cP_{2,0}$. 
Then the posterior $\cP_t$ is also independent across arms:
    $\cP_t = \cP_{1,t}\times \cP_{2,t}$,
where each per-arm posterior $\cP_{a,t}$ is determined by the respective per-arm prior $\cP_{t,0}$ and the history of arm $a$. The basic version is that each $\cP_{a,0}$, $a\in[2]$ is a uniform distribution on $[0,1]$. We allow more general prior beliefs given by Beta distributions: each $\cP_{a,0}$ is a Beta distribution with parameters $\alpha_a,\beta_a\in \N$.

The basic behavior is that each agent $t$ chooses an arm $a\in [2]$ with largest posterior mean reward, $\EE\sbr{\mu_a\mid\hist_t}$. Such agents are called \emph{Bayesian-unbiased}, and the corresponding algorithm is called \emph{Bayesian-greedy}. (This is well-defined even if the beliefs are not independent.)

More generally, we allow a Bayesian version of $\eta$-confident agents, defined as follows. Each agent $t$ maps its posterior $\cP_{a,t}$, $a\in[2]$ to the index $\indx_{a,t}$ for arm $a$, and chooses an arm with a largest index (breaking ties independently and uniformly at random). For unbiased agents, $\indx_{a,t}$ is the posterior mean reward. More generally, we allow
\begin{align}\label{eq:ind-Bayes}
\indx_{a,t} \in \sbr{Q_{a,t}(\zeta),\; Q_{a,t}(1-\zeta)}
    \quad \text{for each arm $a\in[2]$},
\end{align}
where $Q_{a,t}(\cdot)$ denotes the quantile function of the posterior $\cP_{a,t}$ and $\zeta\in (0,\nicefrac12)$ is a fixed parameter (analogous to $\eta$ elsewhere). The interval in \refeq{eq:ind-Bayes} is a Bayesian version of $\eta$-confidence intervals. Agents $t$ that satisfy \refeq{eq:ind-Bayes} are called \emph{$\zeta$-Bayesian-confident}.

%The basic version is that $\cP_{a,0}$ all agents believe that the mean reward of each arm is initially drawn from a uniform distribution on $[0,1]$. We allow more general beliefs given by independent Beta distributions. For each arm $a\in[2]$, all agents believe that the mean reward $\mu_a$ is initially drawn as an independent sample from Beta distribution with parameters $\alpha_a,\beta_a\in \N$.

\begin{discussion}
$\zeta$-Bayesian-confident agents subsume Bayesian version of optimism and pessimism, where the index $\indx_{a,t}$ is defined as, resp., $Q_{a,t}(1-\zeta)$ and $Q_{a,t}(\zeta)$, as well as all other behavioral biases discussed in Section~\ref{sec:prelims-cases}. In particular, one can define an inherently ``Bayesian" version of ``moderate probability matching" by projecting the posterior sample $\nu_{a,t}$ (as defined in Section~\ref{sec:prelims-cases}, but starting with arbitrary Beta-beliefs) into the Bayesian confidence interval \eqref{eq:ind-Bayes}.
\end{discussion}

\xhdr{Our results.}
Recall that prior belief $\cP_{a,0}$ for each arm $a\in[a]$ is a Beta distribution with parameters $\alpha_a,\beta_a\in \N$. Our guarantees are driven by parameter
    $M = \max_{a\in[2]} \alpha_a+\beta_a$.
We refer to such beliefs as \emph{Beta-beliefs with strength $M$}. The intuition is that the prior on each arm $a$ can be interpreted as being ``based on" $\alpha_a+\beta_a-2$ samples from this arm.%
\footnote{More precisely, any Beta distribution with integer parameters $(\alpha,\beta)$ can be seen as a Bayesian posterior obtained by updating a uniform prior on $[0,1]$ with $\alpha+\beta-2$ data points.}

Our technical contribution here is that Bayesian-unbiased (resp., $\zeta$-Bayesian-confident) agents are $\eta$-confident for a suitably large $\eta$. The proof is deferred to Appendix~\ref{app:Bayes}.

\begin{theorem}\label{thm:Bayes-unbiased}
Consider a Bayesian agent that holds Beta-beliefs with strength $M\geq 1$.
\begin{itemize}
\item[(a)] If the agent is Bayesian-unbiased, then it is $\eta$-confident for
some $\eta = O(M/\sqrt{N_0})$.

\item[(b)] If the agent is $\zeta$-Bayesian-confident, then it is $\eta$-confident for some
    $\eta = O\rbr{M/\sqrt{N_0} + \ln(1/\zeta)}$.
\end{itemize}
\end{theorem}

Recall that such agents are subject to the learning failures derived in Theorems~\ref{thm:interval} and~\ref{thm:LB-ext}.

\begin{discussion}
We allow arbitrary Beta-beliefs, possibly completely unrelated to the actual mean rewards. In fact, the theorem holds even if different agents $t$ have different prior beliefs with strength $M_t\leq M$. If $\zeta$ and $M$ are constants relative to $T$, the resulting $\eta$ is constant, too. Our guarantee is stronger if the beliefs are weak (\ie $M$ is small) or are ``dominated" by the initial samples, in the sense that $N_0>\Omega(M^2)$.
\end{discussion}

%% file: sec-fully-Bayesian.tex
%\section{Bayesian model with arbitrary correlated priors}
\section{Bayesian agents in Bayesian bandits}
\label{sec:priors}

In this section, we consider Bayesian agents (as defined in the previous section) in the environment of \emph{Bayesian bandits}, \ie when the mean rewards $(\mu_1,\mu_2)$ are actually drawn from the prior $\cP$. Either arm could be realized as the good arm (else, the problem is trivial). We focus on the standard version of Bayesian agents: Bayesian-unbiased agents without any initial data (\ie $N_0=0$).
Put differently, we consider the Bayesian-greedy algorithm in Bayesian bandits.
\footnote{When both arms have the same posterior mean reward, a tie can be broken arbitrarily.}
We are interested in \emph{Bayesian probability} and \emph{Bayesian regret}, \ie resp., probability and regret in expectation over the prior. In contrast with Section~\ref{sec:bayes}, we allow the prior to be correlated across arms (although our final guarantees are strongest for the case of independent priors).

%We analyze Bayesian versions of unbiased agents and $\eta$-confident agents, interpreting them as (frequentist) $\eta'$-confident agents defined above (with slightly larger parameter $\eta'$). We restrict our analysis to Beta distributions that are independent across arms. The details are in Section~\ref{sec:bayes}.

%Compared to Section~\ref{sec:bayes}, the benefit is that we allow arbitrary priors, possibly correlated across the two arms.

Our main technical argument focuses on a \emph{weak} learning failure when the agents never choose an arm with the smaller  prior mean reward (which may or may not be the best arm). 
%as opposed to an arm with the largest \emph{realized} mean reward, which is not necessarily the same arm).
Our guarantee for the Bayesian probability of such failure is particularly clean: it does not depend on the prior, other than through the \emph{prior gap} $\pgap:=\EXP[\mu_1-\mu_2]$, and does not contain any hidden constants.

% On the other hand, the guarantees here are only in expectation over the prior, whereas the ones in Section~\ref{sec:bayes} hold for fixed $\mu_1,\mu_2$.

%Moreover, our results are restricted to Bayesian-unbiased agents. We do not explicitly allow initial samples (\ie we posit $N_0=0$ here), because they are implicitly included in the prior.

\begin{theorem}\label{thm:BG}
Suppose the pair $(\mu_1,\mu_2)$ is initially drawn from some Bayesian prior $\cP$ with prior gap
    $\pgap:=\EXP[\mu_1-\mu_2]>0$,
and there are no initial samples (\ie $N_0=0$). Assume that all agents are Bayesian-unbiased, with beliefs given by $\cP$. Then with Bayesian probability at least $\pgap$, the agents never choose arm $2$.
\end{theorem}

\begin{proof}
W.l.o.g., assume that agents break ties in favor of arm $2$.

In each round $t$, the key quantity is
    $Z_t = \EXP[ \mu_1-\mu_2 \mid \hist_t ]$.
Indeed, arm $2$ is chosen if and only if $Z_t\leq 0$. Let $\tau$ be the first round when arm $2$ is chosen, or $T+1$ if this never happens. We use martingale techniques to prove that
\begin{align}\label{eq:thm:BG-Bayes-OST}
\EXP[Z_\tau] = \EXP[\mu_1-\mu_2].
\end{align}

We obtain \refeq{eq:thm:BG-Bayes-OST} using the optional stopping theorem. We observe that $\tau$ is a stopping time relative to
    $\cH = \rbr{\hist_t:\, t\in [T+1]}$,
and $\rbr{ Z_t: t\in [T+1]}$ is a martingale relative to $\cH$.
\footnote{The latter follows from a general fact that sequence
    $\EXP[X\mid \hist_t]$, $t\in [T+1]$
is a martingale w.r.t. $\cH$ for any random variable $X$ with $\EXP\sbr{|X|}<\infty$. This sequence is known as \emph{Doob martingale} for $X$.}
	% \kbedit{
	% 	Specifically,
	% 	\begin{align*}
	% 		\Ex{Z_{t+1} \mid \hist_t}
	% 		&=
	% 		\Ex{\Ex{\mu_1 - \mu_2 \mid \hist_{t+1}} \mid \hist_t}
	% 		&\EqComment{Definition of $Z_{t+1}$}
	% 		\\&=
	% 		\Ex{\mu_1 - \mu \mid \hist_t}
	% 		&\EqComment{Iterated expectation}
	% 		\\&=
	% 		Z_{t}
	% 		&\EqComment{Definition of $Z_{t}$}
	% 		.
	% 	\end{align*}
	% }
The optional stopping theorem asserts that $\EXP[Z_\tau] = \EXP[Z_1]$ for any martingale $Z_t$ and any bounded stopping time $\tau$. \refeq{eq:thm:BG-Bayes-OST} follows because
    $\EXP[Z_1] = \EXP[\mu_1-\mu_2]$.

On the other hand, by Bayes' theorem it holds that
\begin{align}\label{eq:thm:BG-Bayes}
\EXP[Z_\tau]
    = \Pr{\tau\leq T ]\,\EXP[ Z_\tau \mid \tau\leq T }
        + \Pr{ \tau>T ]\,\EXP[ Z_\tau \mid \tau>T }
\end{align}
Recall that $\tau\leq T$ implies that arm $2$ is chosen in round $\tau$, which in turn implies that $Z_\tau \leq 0$. It follows that
    $\EXP[ Z_\tau \mid \tau\leq T ]\leq 0$.
Plugging this into \refeq{eq:thm:BG-Bayes}, we find that
\[ \EXP[\mu_1-\mu_2] = \EXP[Z_\tau] \leq \Pr{\tau>T}.  \]
And $\{\tau>T\}$ is precisely the event that arm $2$ is never chosen.
\end{proof}

We obtain a 0-sampling failure when a weak learning failure happens and arm $2$ is in fact the best arm. We lower-bound the probability of that happening, leading to $\Omega(T)$ Bayesian regret.
%As a corollary, we derive a 0-sampling failure, leading to $\Omega(T)$ Bayesian regret. Specifically, the agents start out playing arm $1$ (because $\pgap>0$), and never try arm $2$ \emph{when it is in fact the best arm}.
The cleanest version of this result assumes that the prior has a probability density function which is uniformly bounded away from $0$.
%, \ie $f_\cP(\mu_1,\mu_2)\geq \minPDF$ for some $\minPDF>0$ and all $\mu_1,\mu_2\in[0,1]$.
%This is a (much) more general family of priors compared to independent Beta-priors allowed in Section~\ref{sec:bayes}.

\begin{corollary}\label{cor:BG}
In the setting of Theorem~\ref{thm:BG}, suppose the prior $\cP$
%is independent across arms and
has a probability density function (p.d.f.) which is uniformly lower-bounded by some $\minPDF>0$. Then
\begin{align}\label{eq:cor:BG}
\EXP\sbr{\regret(T)} \geq c_\cP\cdot T,\quad
\text{for some $c_\cP> 0$ determined by $\cP$}.
\end{align}
\noindent Specifically, recall that
    $\pgap:=\EXP[\mu_1-\mu_2]>0$
is the prior gap. Pick any $\alpha>0$ such that
    $\Pr{\mu_1\geq 1-2\,\alpha} \leq \pgap/2$.
Then one can take
\begin{align}\label{eq:cor:BG-basic-cP}
c_\cP = \alpha\;\pgap \; \Lambda_{\cP}/2,
\text{ where }
\Lambda_{\cP} := \inf_{\nu\in [0,1]} \Pr{\mu_2>1-\alpha \mid \mu_1=\nu}>\alpha \minPDF.
\end{align}
\end{corollary}

\begin{remark}\label{rem:cont-prob}
The conditional probability in \refeq{eq:cor:BG-basic-cP} is defined via the joint density of $(\mu_1,\mu_2)$,
%$\IND\cbr{\mu_2>1-\alpha}$ and $\mu_1$ (which exists because the prior $\cP$ has a p.d.f.).This conditional probability
and is well-defined because, by assumption, the density of $\mu_1$ strictly positive everywhere.
%, whereby
%    $\Pr{\mu_2>1-\alpha \mid \mu_1}$
%is a random variable which has a probability density function $g:[0,1]\to\R$. So,
%    $\Lambda_{\cP}$ is defined as $\inf_{\nu\in [0,1]} g(\nu)$.
%We use a similar shorthand in the proof below.
\end{remark}

While Corollary~\ref{cor:BG} is very general in the abstract formulation of \refeq{eq:cor:BG}, the ``failure strength" is limited for some correlated priors due the infimum in \eqref{eq:cor:BG-basic-cP}. Essentially, the prior must assign a substantial probability to $\mu_2$ being very large conditional on every realization of $\mu_1$.
%rather than conditional on the entire event that $\mu_1$ is sufficiently small.

\begin{proof}[Proof of Corollary~\ref{cor:BG}]
Fix $\alpha$ as specified. Consider the following three events: event
    $\GoodSmall:= \cbr{\mu_1 <1-2\alpha}$
that $\mu_1$ is upper-bounded, event
    $\BadLarge:= \cbr{\mu_2 > 1-\alpha}$
that $\mu_2$ is lower-bounded, and event $F$ that arm $2$ is never chosen. We are interested in the intersection of these events. Then each round contributes $\mu_2-\mu_1\geq \alpha$ to regret, so that
    $\EXP\sbr{\regret(T) \mid \GoodSmall, \BadLarge, F} \geq \alpha\,T$.

Next, we lower-bound
    $\Pr{\GoodSmall, \BadLarge, F}$.
We invoke the p.d.f. to prove that
\begin{align}\label{eq:cor:BG-key}
\Pr{\BadLarge,F \mid \GoodSmall} \geq \Lambda_{\cP}\cdot \Pr{F\mid \GoodSmall}.
\end{align}
Once we have \eqref{eq:cor:BG-key}, we continue as follows:
\begin{align*}
\Pr{\GoodSmall, \BadLarge, F}
    &:= \Pr{\BadLarge,F \mid \GoodSmall} \cdot \Pr{\GoodSmall} \\
    &\geq \Lambda_{\cP}\cdot \Pr{F\mid \GoodSmall} \cdot \Pr{\GoodSmall} \\
    &\geq \Lambda_{\cP}\cdot \Pr{\GoodSmall, F}.
\end{align*}
Finally, by Theorem~\ref{thm:BG} and the choice of $\alpha$ we have
    $\Pr{\GoodSmall, F}\geq \Pr{F} - \Pr{\text{not }\GoodSmall} \geq \pgap/2$.
This yields the claimed regret bound: \refeq{eq:cor:BG} with
$c_\cP = \alpha\;\pgap \; \Lambda_{\cP}/2$.

It remains to prove \refeq{eq:cor:BG-key}. Due to the assumption that the p.d.f. exists and is lower-bounded by $\minPDF>0$, the following Riemann integrals are well-defined:
\begin{align}
\Pr{\BadLarge,F \mid \GoodSmall}
    &= \int_{\GoodSmall}
            \Pr{F,\BadLarge \mid \mu_1} \cdot \Pr{\mu_1 \mid \GoodSmall} \dd \mu_1
            \nonumber\\
    &= \int_{\GoodSmall}
            \Pr{\BadLarge \mid \mu_1} \cdot \Pr{F \mid \mu_1}
            \cdot \Pr{\mu_1 \mid \GoodSmall}
            \dd \mu_1
            \label{eq:cor:BG-2}\\
    &\geq \int_{\GoodSmall}
            \Lambda_{\cP}\cdot \Pr{F \mid \mu_1}
            \cdot \Pr{\mu_1 \mid \GoodSmall}
            \dd \mu_1
            \label{eq:cor:BG-3}\\
    &= \Lambda_{\cP}\cdot \Pr{F\mid \GoodSmall} \nonumber.
\end{align}
Here, \eqref{eq:cor:BG-2} uses the fact that event $F$ is determined by reward realizations of arm $1$, and therefore is conditionally independent with $\BadLarge$ given $\mu_1$, and \eqref{eq:cor:BG-2} invokes the definition of $\Lambda_{\cP}$ as a lower bound on
    $\Pr{\BadLarge \mid \mu_1}$.
This completes the proof.
\end{proof}

We also provide versions of Corollary~\ref{cor:BG} without assuming the existence of a density function: (a) a simpler version for independent priors and (b) a similar version if $\mu_1$ has a finite support set.
%, and (c) the most general version for arbitrary correlated priors via Lebesgue intergals.

%All three are proved by substituting a suitable argument for \refeq{eq:cor:BG-key}.

\begin{corollary}\label{cor:BG-general}
In the setting of Theorem~\ref{thm:BG}, suppose
    $\Pr{\mu_1=1}<\pgap/2$.
%is independent across arms and satisfies $\Pr{\mu_1=1}=0$.
%Consider independent priors such that
%    $\Pr{\mu_1=1}<(\mu_1^0-\mu_2^0)/2$.
Pick any $\alpha>0$ such that
    $\Pr{\mu_1\geq 1-2\,\alpha} \leq \pgap/2$.
Then
    $\EXP\sbr{\regret(T)} \geq T\cdot \rbr{ \nicefrac{\alpha}{2}\;\pgap \; \Lambda_{\cP}}$,
where $\Lambda_{\cP}$ is as follows:
\begin{OneLiners}
\item[(a)] If the prior $\cP$ is independent across arms, then
%\begin{align}\label{eq:cor:BG-indep}
$\Lambda_{\cP} = \Pr{\mu_2>1-\alpha}$.
%\end{align}

\item[(b)] if $\mu_1$ has a finite support set $M$, then
%\begin{align}\label{eq:cor:BG-finite}
$\Lambda_{\cP} = \min_{\nu\in M} \Pr{\mu_2>1-\alpha \mid \mu_1 = \nu}$.
%\end{align}

%\item[(c)] in general, $\Lambda_{\cP}$ is given via a Lebesgue integral: letting
%    $\cP(M) := \Pr{(\mu_1,\mu_2)\in M}$,
%\begin{align}\label{eq:cor:BG-general}
%\Lambda_{\cP} = \inf_{M\subset [0,1]^2:\; \cP(M)>0} \quad
%    \frac{1}{\cP(M)} \; \int_{(\mu_1,\,\mu_2)\in M}  \Pr{\mu_2>1-\alpha \mid \mu_1} \dd \cP.
%\end{align}
\end{OneLiners}
\end{corollary}

\begin{proof}
Both parts follow from the proof of Corollary~\ref{cor:BG} as spelled out in Section~\ref{sec:priors}, substituting a suitable argument to prove \refeq{eq:cor:BG-key}.

For part (a), \refeq{eq:cor:BG-key} holds, with $\Lambda_{\cP} = \Pr{\BadLarge}$ as specified, because events $\GoodSmall$ and $F$ are determined by the realization of $\mu_1$ and the rewards of arm $1$, and therefore is independent of $\mu_2$. Consequently,
\begin{align*}
\Pr{\BadLarge,F \mid \GoodSmall}
   &= \Pr{F\mid \GoodSmall}\cdot \Pr{\BadLarge \mid \GoodSmall} \\
   &= \Pr{F\mid \GoodSmall}\cdot  \Pr{\BadLarge}.
\end{align*}

For part (b), \refeq{eq:cor:BG-key} holds, with $\Lambda_{\cP}$ as specified, by the same argument as in the proof of Corollary~\ref{cor:BG}, with integrals over $\mu_1\in \GoodSmall$ replaced by sums over $\mu_1\in\GoodSmall\cap M$.
%For part (c), \refeq{eq:cor:BG-key} holds as follows:
%\ascomment{Kiarash, Suho: could you prove this somehow?}
\end{proof}

The finite-support version applies whenever the prior is over finitely many ``states of nature". To ensure linear regret, it suffices to assume that $\mu_1$ takes its largest possible value $\max(M)$ with probability less than $\pgap/2$, and $\mu_2$ can take this value conditional on any feasible realization of $\mu_1$.

The version for independent priors handles arbitrary per-arm priors that admit a probability density function, and more generally arbitrary per-arm priors such that $\Pr{\mu_1=1}<\pgap/2$ and $\Pr{\mu_2>1-\alpha}>0$ for any $\alpha>0$. This is a much more general family of priors compared to independent Beta-priors allowed in Section~\ref{sec:bayes}.

%	\kbcomment{
%Define the events $A$ and $F$ as $A := \cbr{\mu_1<1-2\alpha}$ and
%	$F := \cbr{\text{arm $2$ is never chosen}}$ respectively.
%	Define $INF := \inf_{\mu_1 \in A}\Pr{\cE_2 \mid \mu_1}$
%\begin{align*}
%	\Pr{\cE_1 \cap \cE_2}
%	&=
%	\Pr{A \cap F \cap \cE_2}
%	\\&=
%	\Pr{A} \Pr{F \cap \cE_2 \mid A}
%	\\&=
%	\Pr{A} \Exu{\mu_1 \mid A}{\Pr{F \cap \cE_2 \mid \mu_1}}
%	\\&=
%	\Pr{A} \Exu{\mu_1 \mid A}{\Pr{F\mid \mu_1} \Pr{\cE_2\mid \mu_1}}
%	\\&\ge
%	\Pr{A} \Exu{\mu_1 \mid A}{\Pr{F\mid \mu_1} INF}
%	\\&=
%	\Pr{A} INF \Exu{\mu_1 \mid A}{\Pr{F\mid \mu_1}}
%	\\&=
%	\Pr{A} INF \Pr{F \mid A}
%	= INF \Pr{F \cap A}
%	= INF \Pr{\cE_1}
%\end{align*}
%	}

%% file: sec-K.tex
\section{Learning failures for $K\geq 2$ arms}
\label{sec:K}

We extend most of our negative guarantees to \BSL with $K\geq 2$ arms. The setting from Section~\ref{sec:prelims} (and from Section~\ref{sec:bayes} for Bayesian agents) carries over word-by-word, except now the set of arms is $[K]$ and the initial data consists of $N_0$ samples of each arm. We extend the main result (Theorem~\ref{thm:interval}), its extension to pessimistic agents (Theorem~\ref{thm:pessimis}) and the results on Bayesian agents (Theorems~\ref{thm:Bayes-unbiased},~\ref{thm:BG} and Corollaries~\ref{cor:K-BG-clean},~\ref{cor:K-BG-general}),
see Table~\ref{tab:results-negative-K} for a summary. Our guarantees are flexible, as explained below, and in some ways stronger than for the two-armed case, but we make no additional claims about their optimality. The technical novelty lies in formulating these results; the respective proofs from the two-armed case carry over with minor modifications.

\begin{table}[t]
\centering
\begin{tabular}{l|l|l|l}
\multicolumn{1}{c|}{Mean rewards }
        & \multicolumn{1}{c|}{Beliefs}           & \multicolumn{1}{c|}{Behavior}  & \multicolumn{1}{c}{Result}  \\[.5mm]\hline &&&\\[-1.5ex]
fixed   & ``frequentist"
                            & $\eta$-confident      & Thm.~\ref{thm:K-LB}\\
        & ~~~~~confidence intervals
                            & $\eta_t$-pessimistic  & Thm.~\ref{thm:K-LB-pess} \\
        & Bayesian (independent)
                            & Bayesian-unbiased,        & Cor.~\ref{cor:K-Bayesian} \\
        &                   & ~~~~~$\eta$-Bayesian-confident&        \\
Bayesian (correlated)
        & Bayesian (and correct)
                            & Bayesian-unbiased        & Thm.~\ref{thm:K-BG} \\
        &                   &                       & Cor.~\ref{cor:K-BG-clean},
                                                            \ref{cor:K-BG-general}
\end{tabular}
\caption{Our negative results for $K>2$ arms.}
\label{tab:results-negative-K}
\end{table}

\subsection{Frequentist agents}

We extend the main result (Theorem~\ref{thm:interval}). We recover it as stated when $1$ and $2$ are the two best arms. Moreover, since the gap between the two best arms may be very small or zero, we allow a more general type of failure when the top $m\geq 1$ arms are never chosen. The failure probability deteriorates with $m$, though. On the other hand, it helps to have multiple ``decoy arms" that the agents might switch to, not just arm $m+1$.

\begin{theorem}\label{thm:K-LB}
Consider \BSL with $K\geq 2$ arms, ordered so that
    $\mu_1 \geq \mu_2 \cdots \geq \mu_K$.
Suppose all agents are $\eta$-confident, for some fixed $\eta\geq 0$. Assume \eqref{eq:assn-1} and \eqref{eq:assn-2}. For any two arms $m<n$,
\begin{align}\label{eq:thm:K-LB}
&\Pr{\text{top $m$ arms are never chosen}} \nonumber \\
	&\qquad\qquad \geq\min\cbr{\frac{1}{2},\quad (n-m)\cdot \Omega_c\rbr{\sqrt{(1+\eta)/N_0}}
    \cdot e^{-O_c\rbr{m\rbr{\eta \;+\; N_0\,\Delta_n^2}}}},
\end{align}
where $\Delta_n := \mu_1-\mu_n$ is the gap for arm $n$. Letting $\pfail(m,n,\eta)$ be the right-hand side of \eqref{eq:thm:K-LB},
\begin{align}\label{eq:thm:K-LB-reg}
    \regret(T) \geq \Delta_{m+1}\cdot \pfail(m,n,\eta)\cdot T.
\end{align}
\end{theorem}

\begin{remark}\label{rem:thm:K-LB}
We recover Theorem~\ref{thm:interval} as stated by taking $m=1$ and $n=2$. When applying Theorem~\ref{thm:K-LB} to a particular example, pick arms $m<n$ to maximize the regret bound in \eqref{eq:thm:K-LB-reg}. In particular, one would pick some arm $n$ such that
    $\Delta_n< O\rbr{1/\sqrt{N_0}}$.%
\footnote{This is to mitigates the dependence on $N_0$ in the exponent in  \eqref{eq:thm:K-LB-reg}, like in Corollary~\ref{cor:LB-small-gap}.}
Two simple examples:
\begin{OneLiners}
\item $\mu_1-\Delta = \mu_2 = \cdots = \mu_K$ (\ie one good arm): use $m=1$ and $n=K$.
\item $\mu_1 = \mu_2 = \cdots = \mu_{K-1} = \mu_K+\Delta$ (\ie one bad arm): use $m=K-1$ and $n=K$.
\end{OneLiners}
\end{remark}

\begin{remark}
How does our guarantee scale with $K$? In part, this is a matter of perspective: whether one fixes $N_0$, the per-arm number of initial samples, or one fixes $N_0\cdot K$, the \emph{total} number of initial samples. (We take no stance on this, our guarantee holds either way.) Either way, the scaling with $K$ can be very different depending on a problem instance, as per the two examples above.
\end{remark}

\begin{proof}[Proof Sketch for Theorem~\ref{thm:K-LB}]
Compared to the proof of Theorem~\ref{thm:interval}, the changes are as follows. We apply the anti-concentration argument to each of the top $m$ arms separately, obtaining an analog of \refeq{eq:pf:thm:interval-prob-1}. We need an intersection of these per-arm events (which are mutually independent), hence the factor of $m$ in the exponent in our guarantee \eqref{eq:thm:K-LB}.

The martingale argument is applied separately to each arm $j$, $m<j\leq n$. Each such arm is treated like the worst-case $j=n$. Thus, we obtain $n-m$ failure events similar to $\sampfail_2$, for each arm $j$, each with a guarantee like \eqref{eq:pf:thm:interval-prob-2}. These events are mutually independent, and just one of them suffices to guarantee the overall failure. This is how we get the $n-m$ factor in \eqref{eq:thm:K-LB}.%
\footnote{If $n$ independent events have probability $\geq p$ each, their union has probability  	
$\geq \min\cbr{\frac{1}{2}, np/2}$, see Lemma~\ref{lm:ind-one}.}
%Appendix \ref{app:ind-one} for completeness.
\end{proof}

%	\kbedit{
%		Specifically, we use the following lemma, the proof of which is in the appendix.
%		\begin{lemma}\label{lm:ind-one}
%		  Let $A_1, \dots, A_n$ be independent events, each occurring with probability $\ge p$.
%			The probability that at least one of these events occurs is lower bounded by
%			\begin{math}
%				\min\cbr{\frac{1}{2}, \frac{np}{2}}
%				.
%			\end{math}
%		\end{lemma}
%		\begin{proof}
%			Since the events are independent, the probability of at least one of them occurring is lower bounded by
%			\begin{math}
%				\Pr{\cup_{i=1}^{n} A_i}
%				=
%				1 - (1-p)^{n}
%				\ge 1-e^{-np}
%				,
%			\end{math}
%			where we have used the ineqaulity $1+x\le e^{x}$.
%			If $e^{-np} \le 1/2$, then the claim follows. Otherwise,
%			we have $np < 1$.
%			Using the ineqaulity $e^{-x} \le 1-x/2$ which is valid
%			for $x \le 1$, we obtain
%			$\Pr{\cup_{i=1}^{n} A_i} \ge \frac{np}{2}$, finishing the proof.
%		\end{proof}
%	}

We also derive an extension for pessimistic agents similar to Theorem~\ref{thm:pessimis}. Essentially, the right-hand side of \eqref{eq:thm:K-LB} improves from $\pfail(m,n,\eta)$ to $\pfail(m,n,0)$.

\begin{theorem}\label{thm:K-LB-pess}
In Theorem~\ref{thm:K-LB}, suppose that each agents $t$ is $\eta_t$-pessimistic, for some $\eta_t\leq \eta$. Then
\begin{align}\label{eq:thm:K-LB-pess}
	\Pr{\text{top $m$ arms are never chosen}}
        \geq \pfail(m,n,0).
\end{align}
\end{theorem}

\begin{proof}[Proof Sketch]
Revisit the proof of Theorem~\ref{thm:pessimis}, with the same changes as for  Theorem~\ref{thm:K-LB}.
\end{proof}

\subsection{Bayesian agents on a fixed bandit instance}

To handle agents with Bayesian beliefs on a fixed bandit instance $(\mu_1 \LDOTS \mu_K)$, we note that Theorem~\ref{thm:Bayes-unbiased} considers each arm separately, and therefore extends to $K\geq 2$ arms.

\begin{corollary}\label{cor:K-Bayesian}
\BSL with $K\geq 2$ arms satisfies Theorem~\ref{thm:Bayes-unbiased}, and therefore is subject to the failure derived in Theorem~\ref{thm:K-LB}.
\end{corollary}

\subsection{Bayesian agents in Bayesian bandits}
\label{sec:K-Bayesian}

Consider the ``fully Bayesian" model from Section~\ref{sec:priors}, with arbitrarily correlated belief/prior and the mean rewards drawn according to this prior (and no initial data, $N_0=0$). We focus on  Bayesian-unbiased agents, \ie the Bayesian-greedy algorithm. As in Section~\ref{sec:priors}, we allow ties to be broken arbitrarily. We allow the set of arms (\emph{action set}) to be arbitrary, possibly infinite or even uncountable, denote it $\cA$. The mean rewards of the arms are represented by a \emph{reward function} $\mu:\cA\to [0,1]$, which is initially drawn from a Bayesian prior $\cP$.

%We obtain a very general result: we extend Theorem~\ref{thm:BG} to any given \emph{subset} $S$ of arms that is never chosen. Our result holds for an arbitrary action set $\cA$, possibly infinite or even uncountable, and an arbitrary subset $S\subset \cA$. The mean rewards of the arms are represented by a function $\mu:\cA\to [0,1]$, which is initially drawn from a Bayesian prior.

We obtain a general result, extending the weak failure event in Theorem~\ref{thm:BG} to any given \emph{subset} $S$ of arms that are never chosen. 
%\asmargincomment{refactored the "intro paragraphs".}

\begin{theorem}\label{thm:K-BG}
Suppose the mean rewards $\mu:\cA\to [0,1]$ are initially drawn from some (possibly correlated) Bayesian prior $\cP$, and there are no initial samples (\ie $N_0=0$).
Assume that all agents are Bayesian-unbiased, with beliefs given by $\cP$. Pick any subset of arms $S\subset \cA$. Then
\begin{align}\label{eq:thm:K-BG}
	\Pr{\text{no arm in $S$ is ever chosen}}
        \geq \EXP\sbr{ \mu^*(\cA\setminus S) - \mu^*(S) },
\end{align}
where $\mu^*(S) := \max_{a\in S} \mu(a)$ is the largest (realized) mean reward in $S$.
\end{theorem}

\begin{proof}[Proof Sketch]
In the proof of Theorem~\ref{thm:BG}, replace ``arm 2" with subset $S$, and ``arm 1" with $\cA\setminus S$. More concretely, replace $\mu_2$ with $\mu^*(S)$, and $\mu_1$ with $\mu^*(\cA\setminus S)$.
	% \kbedit{
	% 	As before, since $\tau$ is a stopping time and $Z_t$ is a martinagle w.r.t $\cH$,
	% 	Eq. \eqref{eq:thm:BG-Bayes} holds.
	% 	We again have
	% 	$\Ex{Z_{\tau} \mid \tau \le T} \le 0$
	% 	because if $\tau \le T$, an arm in $S$ was chosen
	% 	at time $\tau \le T$, which implies $Z_{\tau} \le 0$.
	% 	The event
	% 	$\cbr{\tau \le T}$ is now equivalent to no arm in $S$ being chosen, finishing the proof.
	% }
\end{proof}

\begin{remark}
We recover Theorem~\ref{thm:BG} for two arms by taking a singleton set $S$ that consists of the second-best arm. More generally, we obtain a non-trivial bound for any subset $S$ which is ``less promising" than $\cA\setminus S$ according to the prior, in the precise sense given by \refeq{eq:thm:K-BG}. Note that when $S$ gets smaller, the right-hand side of \eqref{eq:thm:K-BG} leverages this via both $\mu^*(\cA\setminus S)$ and $-\mu^*(S)$.
\end{remark}

Theorem~\ref{thm:K-BG} implies $\Omega(T)$ Bayesian regret, like in the case of $K=2$ arms. The cleanest result parallels Corollary~\ref{cor:BG}, focusing on priors that admit a probability density function that is bounded away from $0$. In what follows, let $\cM = [0,1]^{\cA}$ be the set of all possible reward functions $\cA\to[0,1]$. 

\begin{corollary}\label{cor:K-BG-clean}
Consider the setting of Theorem~\ref{thm:K-BG} with finitely many arms. Suppose the prior $\cP$ has a probability density function (p.d.f.) over $\cM$ which is uniformly lower-bounded by some $\minPDF>0$, and it is not the case that all arms have the same prior mean rewards $\EE_{\cP}[\mu(a)]$. Then
\refeq{eq:cor:BG} holds.
\end{corollary}

This follows from a more explicit result stated below. Apart from a version with a p.d.f., we provide a similar result under a finite-support assumption and a simpler result under an independence assumption, akin to Corollary~\ref{cor:BG-general}. All three versions %do not require finitely many arms, and
are stated under a common framing.

\begin{corollary}\label{cor:K-BG-general}
Consider the setting of Theorem~\ref{thm:K-BG}. Let $S\subset \cA$ be any subset of arms for which Theorem~\ref{thm:K-BG} gives a non-trivial guarantee
    $\pfail(S) := \EXP\sbr{\mu^*(\cA\setminus S)-\mu^*(S)} >0$,
and moreover
    $\Pr{\mu^*(\cA\setminus S)=1} < \pfail(S)/2$.
Fix any $\alpha>0$ such that
    $\Pr{\mu^*(\cA\setminus S)\geq 1-2\,\alpha} \leq \pfail(S)/2$.
Then Bayesian regret is at least
\begin{align*}
    \EXP\sbr{\regret(T)} \geq T\cdot \rbr{ \nicefrac{\alpha}{2}\cdot \pfail(S) \cdot \Lambda_\cP(S)},
\end{align*}
where $\Lambda_\cP(S)$ is concerned with event
    $\BadLarge[S]:= \cbr{\mu^*(S) > 1-\alpha}$.
Specifically:
\begin{OneLiners}
\item[(a)]
If $\mu^*(S)$ and $\mu^*(\cA\setminus S)$ are mutually independent,
%If the prior $\cP$ is independent across arms,
then
    $\Lambda_{\cP}(S) = \Pr{\BadLarge[S]}$.

\item[(b)] if $\mu^*(\cA\setminus S)$ has a finite support set $M$, then
\begin{align*}%\label{eq:cor:BG-finite}
    \Lambda_{\cP}(S) = \min_{\nu\in M}
        \Pr{\BadLarge[S] \mid \mu^*(\cA\setminus S) = \nu}.
\end{align*}

\item[(c)] Suppose there are finitely many arms, and the prior $\cP$ has a p.d.f. which is uniformly lower-bounded by some $\minPDF>0$. Then
\begin{align*}
    \Lambda_{\cP}(S) = \inf_{\nu\in [0,1]}
        \Pr{\BadLarge[S] \mid \mu^*(\cA\setminus S) = \nu},
\end{align*}
where the conditional probability is defined via the joint density of
$\mu^*(S)$ and $\mu^*(\cA\setminus S)$.
%More formally,
%    $\Pr{\BadLarge[S] \mid \mu^*(\cA\setminus S)}$
%is a random variable which has a probability density function $g:[0,1]\to\R$. We define
%    $\Lambda_{\cP} = \inf_{\nu\in [0,1]} g(\nu)$.
\end{OneLiners}
\end{corollary}

\begin{proof}[Proof Sketches]
Corollary~\ref{cor:K-BG-general} follows from the proofs of Corollaries~\ref{cor:BG} and~\ref{cor:BG-general} -- which essentially carry over word-by-word if one replaces ``arm 2" with subset $S$, and ``arm 1" with subset $\cA\setminus S$. In particular, one replaces $\mu_2$ with $\mu^*(S)$, and $\mu_1$ with $\mu^*(\cA\setminus S)$.

In Corollary~\ref{cor:K-BG-general}(c), the existence of the joint density of $\mu^*(S)$ and $\mu^*(\cA\setminus S)$ follows by standard arguments, see Appendix~\ref{app:joint_density} for completeness.

Corollary~\ref{cor:K-BG-clean} follows from Corollary~\ref{cor:K-BG-clean}(c) by letting $S$ be an arbitrary subset of arms not containing the arm(s) with the largest prior mean reward, so that $\pfail(S)>0$. Since the p.d.f. for $\cP$ exists and is bounded away from $0$, it follows that a suitable $\alpha$ exists and $\Lambda_{\cP}(S)>0$.
\end{proof}

Like the respective corollaries for two arms, these linear-regret results are very general in the abstract formulation of Corollary~\ref{cor:K-BG-clean}, but the ``failure strength" is limited for some correlated priors due the minimum/infumum in the definition of $\Lambda_{\cP}$. On the other hand, the subset $S$ can be chosen arbitrarily so as to increase the failure strength.

Part (a) avoids the minimum/infimum via the independence assumption. Note that this assumption is only on $\mu^*(S)$ and $\mu^*(\cA\setminus S)$, rather than on individual arms.

%\ascomment{When/if we argue about specific structures such as dynamic pricing, I'd assume finitely many ``states of nature" and use Corollary~\ref{cor:K-BG-clean}(b).}

%%%%%%%%%%%%%%%%%%%%%%%%

%\begin{corollary}\label{cor:K-BG-old}
%In the setting of Theorem~\ref{thm:K-BG}, assume with finitely many arms and the prior $\cP$ that is independent across arms and satisfies $\Pr{\mu^*(\cA)=1}=0$. Let $S\subset \cA$ be any subset of arms for which Theorem~\ref{thm:K-BG} gives a non-trivial guarantee
%    $\pfail(S) := \EXP\sbr{\mu^*(\cA\setminus S)-\mu^*(S)} >0$.
%Pick some $\alpha>0$ such that
%    $\Pr{\mu^*(\cA\setminus S)\geq 1-2\,\alpha} \leq \pfail(S)/2$.
%Then Bayesian regret is at least
%\begin{align*}
%    \EXP\sbr{\regret(T)} \geq T\cdot \rbr{ \nicefrac{\alpha}{2}\cdot \pfail(S) \cdot \Pr{\mu^*(S)>1-\alpha} }.
%\end{align*}
%\end{corollary}

%Compared to the proof of Theorem~\ref{thm:interval}, we use a new trick: we apply the anti-concentration argument to all top $m$ arms at once, rather than the good arm alone. For this, we define a single reward-tape for all these arms,
%    $\rbr{\tape^*_i\in [0,1]:\; i\in[T]}$,
%(rather than $m$ different reward-tapes). Since different arms can have different mean rewards, we need to define the reward-tape slightly more carefully:
%\begin{OneLiners}
%\item    $\tape^*_i$ is drawn independently and uniformly from $[0,1]$,
%\item the $i$-th time some arm $j\leq m$ is chosen, its reward equals
%    $\indic{\tape^*_i<\mu_j}$.
%\end{OneLiners}
%We apply anti-concentration to this reward tape, and obtain an analog of \eqref{eq:pf:thm:interval-prob-1}.

%% file: sec-expts.tex
\section{Numerical examples}
\label{sec:expts}

Let us provide some simple numerical examples to illustrate our main theoretical results. We focus on two arms and investigate the empirical probability of a learning failure.

Our experimental setup is as follows. For a particular algorithm / behavioural type, consider the event $F_t$ that the bad arm is chosen in all rounds between $t$ and the time horizon $T$. We are interested in $\Pr{F_t}$ for all $t$. We re-run the simulation $1000$ times, and plot the fraction of runs for which $F_t$ happens, as a curve over time $t$, henceforth called the \emph{fail-curve}.

We focus on the fundamental regime when agents are homogeneously all $\eta$-optimistic (resp., all $\eta$-pessimistic) for some fixed $\eta\geq 0$. We plot the fail-curves for several representative values of $\eta$, ranging from LCB to greedy to UCB. We consider mean rewards of the form $0.5 \pm \epsilon$, where $\epsilon$ specifies the problem instance and controls the ``gap" between the two arms.

The results are summarized in Figure~\ref{fig:expts} on page \pageref{fig:expts}. For time horizon $T = 1000$, we  consider $\epsilon=0.05$ (``large gap", top of the figure) and $\epsilon=0.01$ (``small gap", middle).  We find significant failures which, as one would expect, get worse as $\eta$ decreases (treating LCBs as negative $\eta$).

We also investigate UCBs with larger $\eta$, and find similar failures, albeit with  smaller probabilities. We increase the time horizon to $T=10,000$ to make the failures more apparent.%
\footnote{The smaller failure probabilities do not appear to be an artifact of the stringent definition of a failure. Indeed, we checked that relaxing the definition of $F_t$ to allow for a few samples of the good arm would not increase the observed failure probabilities by much.}

\begin{figure}[p]
 \centering
    \includegraphics[width=0.85\textwidth]{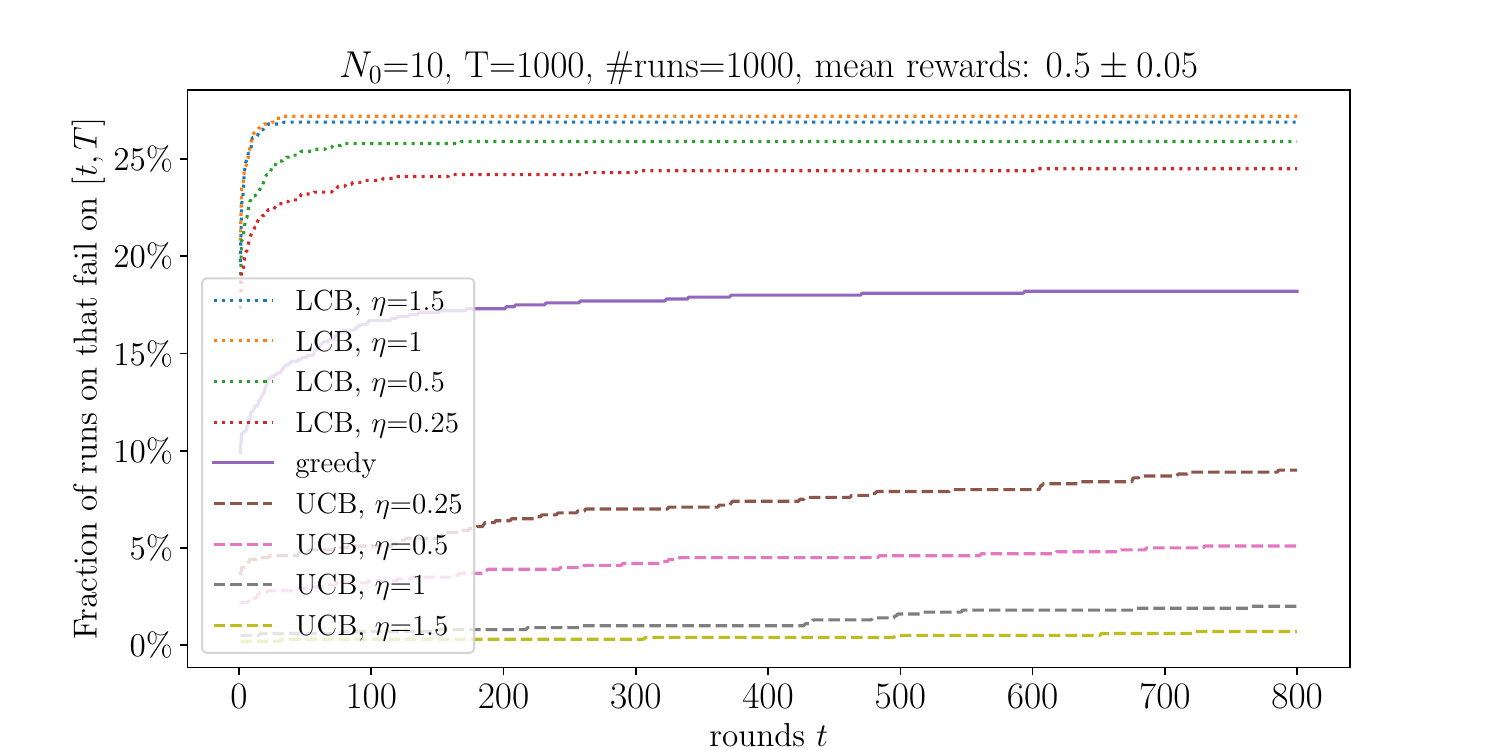}
    \includegraphics[width=0.85\textwidth]{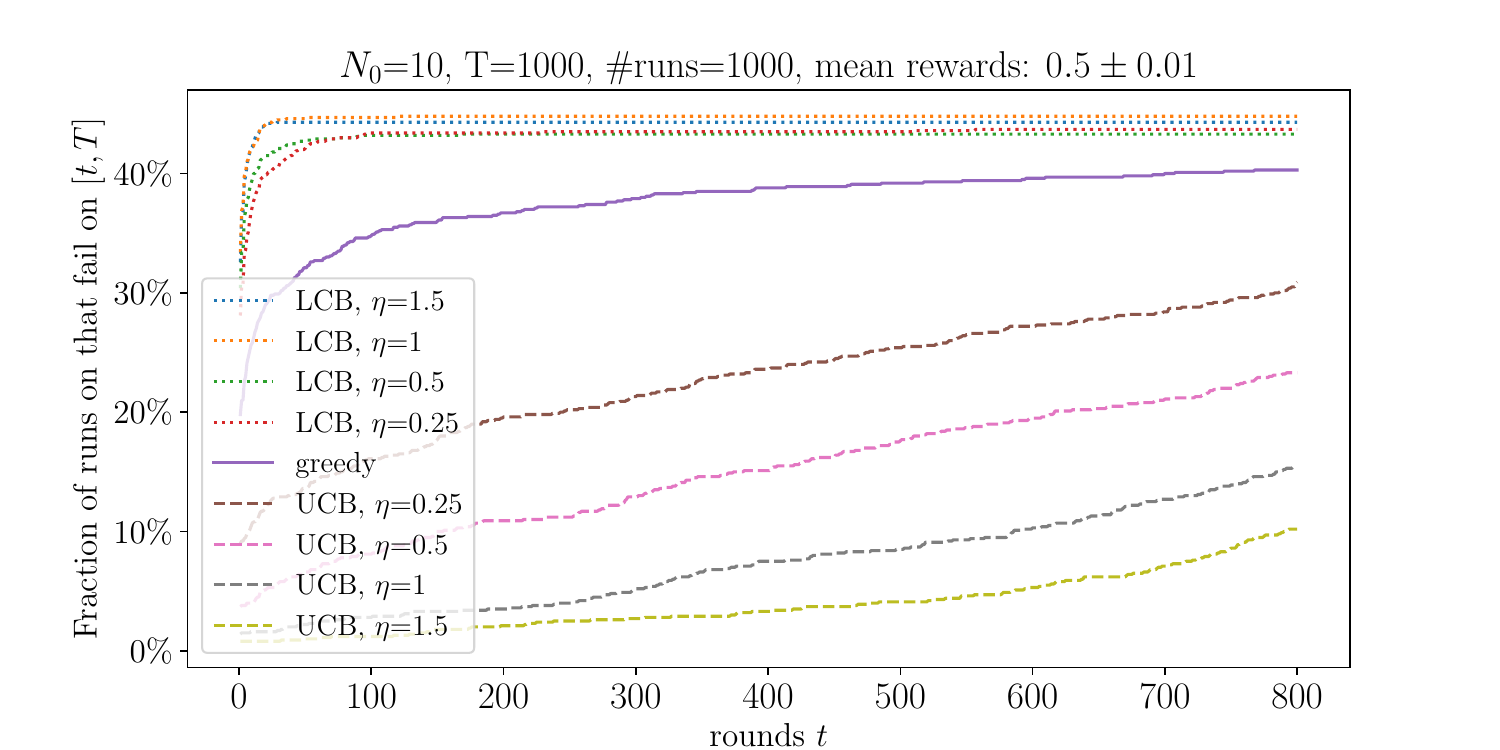}
    \includegraphics[width=0.85\textwidth]{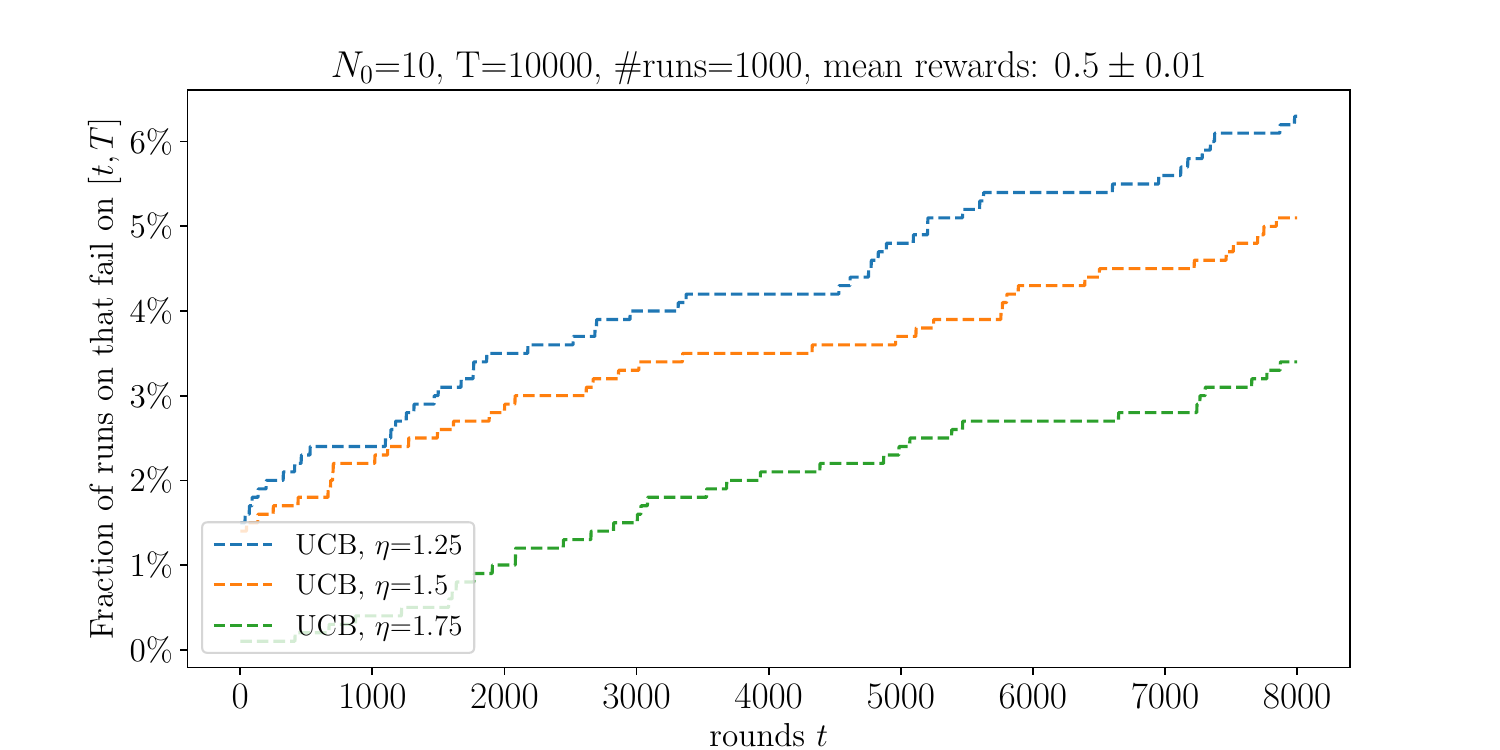}
  \caption{Fail-curves}
  \label{fig:expts}
\end{figure}

%% file: sec-conclusions.tex
\section{Conclusions and open questions}
\label{sec:conclusions}

We examine the dynamics of social learning in a multi-armed bandit scenario, where agents sequentially choose arms and receive rewards, and observe the full history of previous agents. For a range of agents' myopic behavior, we investigate how they impact exploration, and provide tight upper and lower bounds on the learning failure probabilities and regret rates. As a by-product, we obtain the first general results on the failure of the greedy algorithm in bandits.

% Starting from our model, natural open questions concern extending it to multiple arms, particularly with some known correlation that allows inferences across arms.

With our results as a ``departure point", one could study \BSL in more complex bandit models with some known structure of rewards.%
%\sscomment{Minor point re many arm: this might give a sense that we only studied two arms, though we partially addressed many arms.}.
\footnote{\Eg the literature on multi-armed bandits tends to study linear, convex, Lipschitz and combinatorial structures, see the books \citep{Bubeck-survey12,slivkins-MABbook,LS19bandit-book} for background.}
%\asdelete{Alternatively, one could consider Bayesian priors with correlation across arms.}}
%\sscomment{This is also partially addressed -- what about saying ...across arms beyond our results on greedy algorithm?}.}
In particular, the greedy algorithm fails for some structures (\eg our current model) and works well for some others (\eg when all arms have the same rewards), and it is not at all clear what structures would cause failures and/or be amenable to analysis.  Our negative results in Sections~\ref{sec:priors} and~\ref{sec:K-Bayesian} make progress in this direction, as they handle arbitrary ``Bayesian" structures under a full-support assumption. However, these guarantees are restricted to Bayesian bandits (when the mean rewards are drawn according to the agents' prior),  and may be weak or vacuous because of the minimum/infimum in Corollaries~\ref{cor:BG} and~\ref{cor:K-BG-general}.

\xhdr{Follow-up work.}
\cite{Greedy-StructuredMAB-WP25} provide a first general result for \BSL with a known reward structure, focusing on the ``frequentist" greedy algorithm. They characterize whether the greedy algorithm asymptotically ``succeeds" or ``fails" on a given reward structure, in the sense of sublinear vs. linear regret as a function of time. Their characterization is very general: it applies to an arbitrary finite reward structure, and extends to contextual bandits and arbitrary auxiliary feedback. However, their guarantees are quite weak in terms of their dependence on $N_0$ and the reward structure, much like the exponential-in-$N_0$ example from the Introduction.

% A.S.: COMMENTED OUT FOR NOW, SHOULD REVIVE THIS AFTER THE PAPER IS ACCEPTED.
%The most concrete open questions \emph{within} our model concern (i) optimizing the dependence on gap $\Delta$ in all results (focusing on the regime when $\Delta$ is very small) and (ii) extending the results in Section~\ref{sec:priors} beyond unbiased agents.

%% file: app-LB-tools.tex
\section{Probability tools}
% \section{Probability tools: Lemmas~\ref{lm:good_arm_sad} and \ref{lm:bad_arm_happy}}
\label{app:tools}

\subsection{Proof of Lemma~\ref{lm:good_arm_sad}}
\begin{proof}
  We use the following sharp lower bound on the tail probability of binomial distribution.
  \begin{theorem}[Theorem 9 in \cite{zhang2020non}] \label{thm:reverse_hoeffding}
    Let $n \in \N$ be a positive integer and
    let $(X_i)_{i\in [n]}$ be a sequence of i.i.d Bernoulli random variables with parameter $p$.
    For any $\beta > 1$ there exists constants $c_\beta$ and $C_\beta$ that only rely on $\beta$,
    such that for all $x$ satisfying $x \in [0, \frac{np}{\beta}]$ and $x + n(1-p) \ge 1$, we have
    \begin{align*}
      \Pr{\sum_{i=1}^{n} X_i \le np-x} \ge c_\beta e^{-C_\beta nD(p  - \frac{x}{n}|| p)},
    \end{align*}
   where $D(x||y)$ denotes the KL divergence between two Bernoulli random variables with parameters $x$ and $y$.
  \end{theorem}
  We use the above result with $x= n(p-q)$ and $\beta = \frac{1-c}{1-\frac{9}{8}c}$.
  Note that $\beta > 1$ since $c < \frac{1}{2}$.
  % \sscomment{small c. Mention that $c < 8/9$}.
  We first verify that $x, \beta$ satisfy the conditions of the lemma.
  The $x + n(1-p) \ge 1$ condition
  holds by the assumption $n\ge 1/c$:
  \begin{align*}
      x + n(1-p) \ge n(1-p) \ge nc \ge 1.
  \end{align*}
  As for the $x\le \frac{np}{\beta}$ condition,
  by definition of $x$,
  \begin{align*}
    \frac{np}{x} = \frac{np}{n(p-q)} &= \frac{p}{p-q}.
  \end{align*}
  Since $p \le 1-c$ and
  $\frac{p}{p-q}$ is decreasing in $p$ for $p\ge q$,
  we can further bound this with
  \begin{align*}
  \frac{p}{p-q} \ge \frac{1-c}{1-c-q}
    \ge \frac{1-c}{1-c-\frac{c}{8}}
    = \beta,
  \end{align*}
  where the second inequality follows from $q \ge c/8$ and $q < p \le 1-c$, together with the fact that
  $\frac{1-c}{1-c-q}$ is decreasing in $q$ for $q < 1-c$.
  We obtain $x \le \frac{np}{\beta}$ by rearranging.

  Invoking Theorem \ref{thm:reverse_hoeffding} with the given values, we obtain
  \begin{align}
    \Pr{\frac{\sum_{i=1}^n X_i}{n} \le q} \ge c_{\beta}e^{-C_{\beta} nD(q || p)} = \Omega(e^{-O(nD(q || p))})
    .
    \label{eq:jan7_1235}
  \end{align}
  Next, we use the following type of reverse Pinsker's inequality 
  to upper bound $D(q || p)$.
  \begin{theorem}[\cite{gotze2019higher}]
    For any two probability measures $P$ and $Q$ on a finite support $X$, if $Q$ is absolutely continuous with respect to $P$,
    then the their KL divergence $D(Q||P)$ is upper bounded by $\frac{2}{\alpha_P}\delta(Q,P)^2$
    where $\alpha_P = \min_{x \in X} P(x)$ and $\delta(Q,P)$ denotes the total variation distance between $P$ and $Q$.
  \end{theorem}
  Setting $P=\bernoul(p)$ and $Q = \bernoul(q)$,
  we have
  $\alpha_P = \min(p, 1-p)$, and
  $\delta(Q,P) = p - q$
  Therefore,
  since $\min(p, 1-p) \ge c$ by assumption,
  we conclude
  \begin{math}
    D(q || p) \le O((p-q)^2).
  \end{math}
  Plugging this back in Equation \eqref{eq:jan7_1235} finishes 
  % finshes
  the proof.
\end{proof}

\subsection{Proof of Lemma \ref{lm:bad_arm_happy}}
% \begin{enumerate}
%   \item We will use the following
%   \item Ville's inequality.
%   \item Mention that we will use a martingale.
%   \item Define $u$ (Lemma)
%   \item Proof of lemma. Verify that $Z_n$ is a martingale.
% \end{enumerate}
Our proof will rely on the following
doob-style inequality for (super)martingales.
\begin{lemma}[Ville's Inequality~\cite{ville1939etude}]\label{lm:ville}
  Let $(Z_{n})_{n\ge 0}$ be a positive supermartingale with respect to filtration $(\F_{n})_{n\ge 0}$, i.e. $Z_n \ge \Ex{Z_{n+1} | \F_n}$ for any $n \ge 0$.
  Then the following holds for any $x>0$,
  \begin{align*}
    \Pr{\max_{n \ge 0} Z_n \ge x} \le \Ex{Z_0}/x.
  \end{align*}
\end{lemma}
In order to use this result, we will define the martingale
$Z_n := u^{\sum_{i=1}^{n}(X_{i+1} - q)}$ for a suitable choice of $u$
as specified by the following lemma.
\begin{lemma}\label{lm:u_def}
  Let $c$ be an absolute constant.
  For any $p \in (c, 1-c)$ and $q \in (0, p)$, there exists
  a value of $u \in (0, 1)$ such that
  \begin{align}
    (p \cdot u^{1-q} + (1-p) \cdot u^{-q})=1
    \label{eq:u_def}.
  \end{align}
  In addition, $u$ satisfies
  \begin{align}
    p(1-u^{1-q}) \ge \Omega(p-q)
    \label{eq:u_prop}.
  \end{align}
\end{lemma}
\begin{proof}
  To see why such a $u$ exists, define $f(x) = (p \cdot x^{1-q} + (1-p) \cdot x^{-q})$.
  It is clear that $f(1)=1$ and $\lim_{x\to0}f(x) = \infty$ as
  $\lim_{x \to 0} (1-p)x^{-q} = \infty$.
  Furthermore,
  \begin{equation*}
    f'(x) = p \cdot (1-q) \cdot x^{-q} + (1-p) \cdot (-q) \cdot x^{-q-1}
    ,
  \end{equation*}
  which implies
  \begin{equation*}
    f'(1) = p(1-q)  -(1-p)q
    =
    p - q
    > 0
    .
  \end{equation*}
  Therefore, $f(x)$ is decreasing at $x=1$. Since $\lim_{x\to 0} f(x) > f(1)$, this implies that $f(u) = f(1)$
  for some $u \in (0, 1)$, proving Equation \eqref{eq:u_def}.

  We now prove Equation \eqref{eq:u_prop},
  define $x_0$ as $x_0=\frac{(1-p)q}{p(1-q)}$.
  Note that $x_0< 1$ since $p > q$.
  We claim that $u\le x_0$.
  To see why, we first note that $f'(x)$ can be rewritten as
  \begin{align*}
    x^{-q - 1}\left(
      xp(1-q) -(1-p)q
    \right)
    .
  \end{align*}
  It is clear that $f'(x_0) = 0$. Since
  $xp(1-q) - (1-p)q$ is increasing in $x$,
  this further implies
  that
  $f'(x) > 0$ for $x > x_0$.
  Now, if $u > x_0$, then since $f'(x) > 0$ for $x > x_0$,
  we would conclude that $f(u) < f(1)$, which is not possible since
  $f(u)=f(1)=1$. Therefore,
  $u\le x_0$ as claimed.

  We now claim that $x_0^{1-q} \le 1- p + q$.
  This would finish the proof since, together with $u\le x_0$, this would imply
  \begin{align*}
    p(1-u^{1-q}) &\ge p(1-x_0^{1-q})
    \ge p(p - q) = \Omega(p-q),
  \end{align*}
  where for the last equation we have used the assumption $p  \in (c,1-c)$.

  To prove the claim, define $\eps := p-q$.
  We need to show that
  $x_0^{1-q} \le 1-\eps$, or equivalently
  $\ln(x_0) \le \frac{\ln(1-\eps)}{1-q}$.
  By defintion of $x_0$, this is equivalent to
  \begin{align}
    \ln \left(\frac{(1-p)(p-\eps)}{(1-p+\eps)p}\right) \le \frac{1}{1-p+\eps}\ln (1-\eps).
    \label{eq:jan8_1118}
  \end{align}
  Fix $p$ and consider both hand sides as a function of $\eps$.
  Putting $\eps = 0$, both hands side coincide as they both equal 0.
  To prove
  \refeq{eq:jan8_1118}, it suffices to show that as we increase
  $\eps$, the left hand side decreases faster than the right hand side.
  Equivalently, we need to show that the derivative of the LHS
  with respect to $\eps$ is larger than the derivative of the RHS with respect to $\eps$
  for $\eps \le [0, p]$.
  Taking the derivative with respect to $\eps$ on LHS, we obtain
  \begin{align*}
    \frac{d}{d\eps}\left(\ln(1-p) + \ln(p-\eps)  - \ln(1-p+\eps) - \ln (p)\right) = -\frac{1}{p-\eps} -  \frac{1}{1-p+\eps}.
  \end{align*}
  Similarly taking the derivative on RHS we obtain
  \begin{align*}
    \frac{d}{d\eps}\left( \frac{\ln(1-\eps)}{1-p+\eps} \right)  = -\frac{1}{(1-\eps)(1-p+\eps)} -\frac{\ln(1-\eps)}{(1-p+\eps)^2}.
  \end{align*}
  We therefore need to show that
  \begin{align}
    \frac{-1}{1-p+\eps} + \frac{-1}{p-\eps} \le \frac{-1}{(1-p+\eps)(1-\eps)} + \frac{-\ln(1-\eps)}{(1-p+\eps)^2}
    .
    \label{eq:jan8_1119}
  \end{align}
  We note however that
  \begin{align*}
    \frac{-1}{1-p+\eps} + \frac{-1}{p-\eps}
    = \frac{\eps - p  - 1 +  p- \eps}{(1-p+\eps)(1-\eps)}
    = \frac{-1}{(1-p+\eps)(1-\eps)}
    .
  \end{align*}
  Therefore Equation \eqref{eq:jan8_1119} is equivalent to
  \begin{align*}
    \frac{-\ln(1-\eps)}{(1-p+\eps)^2} \ge 0,
  \end{align*}
  which is true since $\eps \in [0, p]$.
  This proves the claim $x_0^{1-q} \le 1-\eps$, finishing the proof.
\end{proof}

We now prove Lemma \ref{lm:bad_arm_happy} using Lemma \ref{lm:ville} and \ref{lm:u_def}.
\begin{proof}[Proof of Lemma \ref{lm:bad_arm_happy}]
  Define the random variable $Y_i$ as $Y_i = X_{i + 1} - q$.
  Note that $Y_i$ takes value $1-q$ with probability $p$ and takes $-q$ with probability $1-p$.
  Set $u$ to be the value specified in Lemma \ref{lm:u_def}.
  For $n\ge 0$, define $Z_n := u^{\sum_{i=1}^n Y_i}$. We first observe that $Z_n$ is a martingale with respect to $Y_1,\ldots, Y_n$ as
  \begin{align*}
    \Ex{Z_{n+1} | Y_{1}, \dots Y_n}
    =
    \Ex{u^{\sum_{i=1}^{n+1}Y_i} | Y_{1}, \dots Y_n}
    &=
    u^{\sum_{i=1}^{n} Y_i} \cdot (p \cdot u^{1-q} + (1-p) \cdot u^{-q})
    \\&= u^{\sum_{i=1}^{n} Y_i}
    =
    Z_{n}.
  \end{align*}
  Since $0<u<1$, this further implies
  \begin{align*}
    \Pr{\forall n \ge 0: \sum_{i=1}^n Y_i \ge q-1 }
    &=1-\Pr{\exists n \ge 0: \sum_{i=1}^n Y_i < q-1}
    \\ &= 1-\Pr{\max_{j \in [n]}\{u^{\sum_{i=1}^j Y_i}  \} \ge u^{q-1}}
    \\ &\ge 1-\frac{\Ex{Z_1}}{u^{q-1}}
    \\ &= 1-u^{1-q}
    ,
  \end{align*}
  where the first inequality follows from Lemma \ref{lm:ville} and the final equality
  follows from $\Ex{Z_1} = \Ex{Z_0} = \Ex{u^{0}} = 1$.

  Since $Y_i$ is a function of $X_{s+1}$, we independently have $X_1=1$
  with probability $p$. Therefore,
  with probability $p(1-u^{1-q})$.
  \begin{align*}
    X_i = 1 \text{ and }
    \forall n \ge 1: \sum_{i=2}^n (X_i-q) \ge q-1,
  \end{align*}
  which further implies $\sum_{i=1}^n(X_i-q) \ge 0$.
  Therefore,
  \begin{align*}
    \Pr{\forall n \ge 1 : \frac{\sum_{i=1}^n X_i}{n} \ge q} \ge p(1-u^{1-q})
    &\ge \Omega(p-q),
  \end{align*}
  where the inequality follows from Equation \eqref{eq:u_prop}.
\end{proof}

\subsection{Union of independent events}
\label{app:ind-one}

%We now state and prove Lemma~\ref{lm:ind-one} which lower bounds the probability of the union of $n$ independent events.

The following result/proof is standard and provided for the sake of completeness.

\begin{lemma}\label{lm:ind-one}
  Let $A_1, \dots, A_n$ be independent events, each occurring with probability $\ge p$.
    The probability that at least one of these events occurs is lower bounded by
    \begin{math}
        \min\cbr{\frac{1}{2}, \frac{np}{2}}
        .
    \end{math}
\end{lemma}
\begin{proof}
    Since the events are independent, the probability of at least one of them occurring is lower bounded by
    \begin{math}
        \Pr{\cup_{i=1}^{n} A_i}
        =
        1 - (1-p)^{n}
        \ge 1-e^{-np}
        ,
    \end{math}
    where we have used the ineqaulity $1+x\le e^{x}$.
    If $e^{-np} \le 1/2$, then the claim follows. Otherwise,
    we have $np < 1$.
    Using the ineqaulity $e^{-x} \le 1-x/2$ which is valid
    for $x \le 1$, we obtain
    $\Pr{\cup_{i=1}^{n} A_i} \ge \frac{np}{2}$, finishing the proof.
\end{proof}

\subsection{Joint density function in Corollary~\ref{cor:K-BG-general}(c)}
\label{app:joint_density}

Recall that Corollary~\ref{cor:K-BG-general}(c) requires the existence of the joint density of $\mu^*(S)$ and $\mu^*(\cA\setminus S)$. While this follows from standard arguments, we provide the proof for completeness.

\begin{lemma}\label{lm:joint_density}
Let $\mu = (\mu_1, \dots, \mu_K)\in [0,1]^K$ be a random vector with a joint p.d.f. (probability density function) $f$. Fix subset $S \subseteq [K]$, let $\mu^*(S):= \max_{i \in S} \mu_i$.
Then random variables $X = \mu^*(S)$ and $Y=\mu^*([K] \setminus S)$ have a joint p.d.f.
\end{lemma}

\newcommand{\setIJ}{\mathtt{IJ}}

\begin{proof}
%Let $\mu_{-i, -j} \in \R^{K - 2}$ denote the vector obtained by removing coordinates $i\neq j$ from $\mu$.
%For each $x,y\in [0,1]$, define the set
%  \begin{align*}
%    B_{i, j}(x, y) := \cbr{
%      \mu_{-i, -j}:
%      % \mu_i = x, \mu_j = y,
%      \max_{i' \in S} \mu_{i'} = x,
%      \max_{j' \in [K] \setminus S} \mu_{j'} = y
%    } \subset [0,1]^{K - 2}.
%  \end{align*}
% Given a vector $\nu\in \R^K$, let $\nu_{-i, -j} \in \R^{K - 2}$ denote the vector obtained by removing coordinates $i\neq j$ from $\nu$.

Fix indices $i \in S$, $j \in [K] \setminus S$. Given a vector $\nu\in \R^K$, let
    $\nu_{-i, -j} \in \R^{K - 2}$
denote the vector obtained by removing coordinates $i$ and $j$ from $\nu$. For each $x,y\in [0,1]$, define the set
  \begin{align*}
    B_{i, j}(x, y) := \cbr{
      \nu_{-i, -j}:
      % \mu_i = x, \mu_j = y,
      \nu\in[0,1]^K,\;
      \max_{i' \in S} \nu_{i'} \le x, \;
      \max_{j' \in [K] \setminus S} \nu_{j'} \le y
    } \subset [0,1]^{K - 2}.
  \end{align*}

For shorthand, write
    $\nu = \rbr{\nu_{-i, -j}; \nu_i,\nu_j}$
and
    $\setIJ := S\times \rbr{[K]\setminus S}$
in what follows.

We prove $f_{X,Y}$ defined below is the p.d.f. for $(X, Y)$:
  \begin{align*}
      f_{X, Y}(x, y) :=
      \sum_{(i,j)\in\setIJ}\quad
      \int_{\nu_{-i, -j} \in B_{i, j}(x, y)} f{\rbr{\nu_{-i, -j}; x,y}} \, \dd \nu_{-i, -j}.
  \end{align*}
 % \ascomment{Is this a good notation? Then, pls use it elsewhere. }
 % \ascomment{Also, is it OK to integrate over $\mu$ given that we defined it as a random vector? Perhaps better to replace $\mu$ with $\nu$ throughout?}
 %  \kbcomment{

%where we set $\mu_i=x$ and $\mu_j=y$ when calculating $f(\mu)$ in the above integral and only integrate with respect to $\mu_{-i, -j}$.

More formally, we need to prove that for any $x',y'\in [0,1]$,
\begin{align}\label{eq:pf:lm:joint_density}
  \Pr{X \le x', Y \le y'}
        = \int_{x \le x'} \int_{y \le y'} f_{X, Y}(x, y) \, \dd y \, \dd x.
\end{align}

Fix $x',y'\in [0,1]$. We are interested in the event
    $A := \cbr{X \le x', Y \le y'}$.
Define event
  \begin{align*}
    A_{i, j}: = A \cap
    \cbr{
      \mu_i = \mu^*(S), \mu_j = \mu^*([K] \setminus S)
    }, \quad (i,j)\in\setIJ.
  \end{align*}
Note that $A = \cup_{(i,j)\in\setIJ} A_{i, j}$. Moreover, the intersection
  $A_{i, j} \cap A_{i', j'}$
has zero Borel measure whenever $(i, j) \ne (i', j')$. It follows that
  % It follows that
  % \begin{align*}
  %   \Pr{A} = \sum_{i, j} \Pr{A_{i, j}}
  %   .
  % \end{align*}
  % In the above equation, and throughout the rest of the proof, the sum $\sum_{i, j}$ is over all
  % $i \in S$ and $j \in [K] \setminus S$.
  % By definition however,
  % \begin{align*}
  %   A_{i, j} = \cbr{
  %     \mu_i =
  %   }
  % \end{align*}
% \ascomment{In the eqns below, do you always have $i \in S$ and $j \in [K] \setminus S$? Say this explicitly.}
  \begin{align*}
    \Pr{A}
    &=
    \sum_{(i,j)\in\setIJ} \Pr{A_{i, j}}
    \\&=
    \sum_{(i,j)\in\setIJ}
    \int_{\nu_i \le x'} \int_{\nu_j \le y'}
    \int_{\nu_{-i, -j} \in B_{i, j}(\nu_i, \nu_j)} f(\nu) \dd \nu
    \\&=
    \sum_{(i,j)\in\setIJ}
    \int_{x \le x'} \int_{y \le y'}
    \int_{\nu_{-i, -j} \in B_{i, j}(x, y)} f(\nu_{-i, -j}; x, y) \dd\nu_{-i, -j} \, \dd y\, \dd x
    \\&=
    \int_{x \le x'} \int_{y \le y'}
    \rbr{
    \sum_{(i,j)\in\setIJ}
    \int_{\nu_{-i, -j} \in B_{i, j}(x, y)} f(\nu_{-i, -j}; x, y) \dd\nu_{-i, -j}
    }\, \dd y\, \dd x
    .\
  \end{align*}
  % where, in the last two integrals, we set $\mu_i = x$ and $\mu_j= y$ when calculating $f(\mu)$.
  % \kbedit{where, the sum $\sum_{i, j}$ is over all $i \in S$ and $j \in [K] \setminus S$.}
\refeq{eq:pf:lm:joint_density} follows by definition of $f_{X, Y}$, completing the proof.
\end{proof}

%% file: freq_UB_proof.tex
\section{Proof of Lemma \ref{lm:failure_total_simplify}}
  We assume without loss of generality that $\eta > 2$.
  If $\eta \le 2$,
  the Lemma's statement can be made vacuous  using large enough constants
  in $O$.
  In addition, for mathematical convenience, we will assume that the tape for each arm is infinite,
  % The assumption has no effect on the behaviour 
  even though the entries after $T$ will never actually be seen by any of the agents.

  For each arm $a$, we first separately consider each interval of the form $[n, 2n]$ and
  bound the probability that $\ucbTape_{a, i}$ deviates too much from 
  $\mu_{a}$ for $i\in [n, 2n]$.
  While this can be done crudely by applying a union bound over all $i$,
  we use the following maximal inequality.
  \begin{lemma}[Eq. (2.17) in \cite{hoeffding1963probability}]
    Given
    a sequence of i.i.d. random variables $(X_i)_{i \in [n]}$ in $[0, 1]$ such that
    $\Ex{X_i} = \mu$, the inequality states that
    for any $x > 0$,
    \begin{align*}
      \Pr{\exists i \in [n]: \abs{\sum_{j=1}^{i} \rbr{X_{j} - \mu}} > x} \le 2e^{-\frac{2x^2}{n}}.
    \end{align*}
  \end{lemma}
  % ~\cite{hoeffding1963probability}.
  % using the \emph{maximal deviation} inequality.
  Focusing on some interval of the form $[n, 2n]$ for $n\in \N$,
  and applying this inequality to the reward tape of arm $a$,
  we conclude that
  \begin{align}
    \Pr{\exists i \in [n, 2n]: \abs{\muTape_{a, i} - \mu_{a}} \ge x}
    \le O(e^{-\Omega(nx^2)}).
    \label{eq:jan26_1934}
  \end{align}

  Define $f := \ceil{64\eta/\Delta^2}$. We note that
  $f = \Theta(\eta/\Delta^2)$ given the assumption $\eta > 2$.
  In order to bound $\Pr{\clean_2^{\eta}}$, we will
  apply this inequality to each interval
  $[n, 2n]$ for $n\ge f$,
  and take a union bound.
  Formally,
  \begin{align*}
    % \Pr{\forall i\ge 64\eta/\Delta^2 : \muTape_{2, i} \le \mu_2 + \Delta/8} 
    1-\Pr{\clean_2^{\eta}}
    &\le
    \Pr{\exists i\ge f: \muTape_{2, i} > \mu_2 + \Delta/8} 
    &\EqComment{Since $\sqrt{\eta/i}\le \Delta/8$ for $i\ge f$}
    \\&\le
    \sum_{r = 0}^{\infty}
    \Pr{\exists i \in [f2^{r}, f2^{r+1}] : \muTape_{2, i} > \mu_2 + \Delta/8} 
    &\EqComment{Union bound}
    \\&\le
    O\rbr{
    \sum_{r = 0}^{\infty}
    e^{-\Omega(\eta2^{r})}
  }
    &\EqComment{By \refeq{eq:jan26_1934}}
    \\&\le
    O\rbr{
      \sum_{r = 0}^{\infty}
      e^{-\Omega(\eta(r + 1))}
    }
    &\EqComment{Since $2^r \ge r+1$ for $r\in \N$}
    \\&=
    O(\frac{1}{e^{\Omega(\eta)} -1})
    &\EqComment{Sum of geometric series}
    \\&\le
    O(e^{-\Omega(\eta)})
    &\EqComment{By $\eta > 2$}
  \end{align*}

  In order to bound
  $\Pr{\clean_1^{\eta}}$,
  we separately handle the intervals
  $n< f$ and $n \ge f$.
  For $n \ge f$, repeating the same argument
  as above for arm 1 implies
  \begin{align*}
    % 1-\Pr{\clean_2^{\eta}}
    % &\le
    % \sum_{r = 0}^{\infty}
    % \Pr{\forall i \in [f2^{r}, f2^{r+1}] : \muTape_{2, i} \le \mu_2 + \Delta/8} 
    % \Pr{\forall i\ge f : \muTape_{1, i} \ge \mu_1 - \Delta/2} \ge 1-O(e^{-\Omega(\eta)}).
    \Pr{\exists i\ge f: \muTape_{1, i} < \mu_1 - \Delta/8} 
    \le 
    O(e^{-\Omega(\eta)}).
  \end{align*}
  For $n< f $, we use
  a modified argument that utilizes the extra $\sqrt{\eta/i}$ term in $\ucbTape_{1, i}$.
  Instead of bounding the probability $\muTape_{1, i}$ having deviation $\Delta/8$,
  we bound the probability that it deviates by $\sqrt{\eta/i}$.
  This results in a marked improvement because
  $\sqrt{\eta/i}$ increases as we decrease $i$.
  Formally,
  \begin{align*}
    &\Pr{\exists i \in [1, f]: \muTape_{1, i} < \mu_1 - \sqrt{\eta/i}} 
    \\&\le
    \sum_{r = 0}^{\ceil{\log(f)}}
    \Pr{\exists i \in [2^{r}, 2^{r+1}]: \muTape_{1, i} < \mu_1 - \sqrt{\eta/i}} 
    &\EqComment{Union bound}
    \\&\le
    \sum_{r = 0}^{\ceil{\log(f)}}
    \Pr{\exists i \in [2^{r}, 2^{r+1}]: \muTape_{1, i} < \mu_1 - \sqrt{\eta/2^{r + 1}}} 
    &\EqComment{By assumption on $i$}
    \\&\le
    O\rbr{
      \sum_{r = 0}^{\ceil{\log(f)}}
      e^{-\Omega(\eta)}
    }
    &\EqComment{By \refeq{eq:jan26_1934}}
    \\&=
    O(\ceil{\log(f)}e^{-\Omega(\eta)})
    .
  \end{align*}
  Finally, we note that
  since $\eta > 2$,
  \begin{align*}
    \ceil{\log(f)}
    \le O(1 + \log(f))
    =O(1 + \log(\eta) + \log(1/\Delta))
    .
  \end{align*}
  This implies \refeq{eq:lm:failure_total-1}
  because $O(\log(\eta)e^{-\Omega(\eta)})$
  can be rewritten as $O(e^{-\Omega(\eta)})$ by changing the constant behind $\Omega$.

%% file: app-bayesian.tex
% !TEX root =  main.tex
\section{Proof of Theorem \ref{thm:Bayes-unbiased}}
\label{app:Bayes}
In this section, we prove Theorem \ref{thm:Bayes-unbiased}.
We first briefly review some properties of the beta distribution.
Throughout the section,
we consider a beta distribution with parameters $\alpha, \beta$.
\newcommand{\fbin}{F^{B}}
\newcommand{\fbeta}{F^{beta}}
\begin{lemma}[Fact 1 in \citet{Shipra-colt12}]
  Let $\fbin_{n, p}$ denote the CDF of the binomial distribution with paramters
  $n, p$ and $\fbeta_{\alpha, \beta}$ denote the CDF of the beta distribution.
  Then,
  \begin{align*}
    \fbeta_{\alpha, \beta}(y) = 1-\fbin_{\alpha + \beta - 1, y}(\alpha -1)
  \end{align*}
  for $\alpha, \beta$ that are positive integers.
\end{lemma}
Using Hoeffding's inequality for concentration of the
binomial distribution, we immediately obtain the following corollary.
\begin{corollary}\label{cor:concentration_beta}
  Define
  \begin{math}
    \rho_{\alpha, \beta} := \frac{\alpha-1}{\alpha + \beta-1}.
  \end{math}
  If $X$ is sampled from the beta distribution with parameters $(\alpha, \beta)$,
  \begin{align*}
    \Pr{\abs{X - \rho_{\alpha, \beta}} \le y} \le 2e^{-(\alpha + \beta -1)y^2}
    .
  \end{align*}
  In addition, letting
  $Q(.)$ denote the quantile function of the distribution,
  % The $\zeta$-quantiles of the beta distribution fall in the range
  \begin{align*}
    [Q(\zeta), Q(1-\zeta)] \subseteq
    \sbr{\rho_{\alpha, \beta} - \sqrt{\frac{\ln(2/\zeta)}{\alpha + \beta  -1}}, \rho_{\alpha, \beta} + \sqrt{\frac{\ln(2/\zeta)}{\alpha + \beta  -1}}},
  \end{align*}
\end{corollary}

Let $\alpha_{a, n}, \beta_{a, n}$ denote the posterior distribution after observing
$n$ entries of the tape for arm $a$.
Note that since we are assuming independent priors,
the posterior for each arm is independent of the seen rewards of the other arm.
Define
$M_{a, n} := \alpha_{a, n} + \beta_{a, n}$.
We note that by definition, $\alpha_{a, 0}, \beta_{a, 0}$ coincide with the prior $\alpha_a, \beta_a$.
We analogously define $M_{a} := \alpha_{a} + \beta_{a}$.
Define
$\rho_{a, n} := \frac{\alpha_{a, n} -1}{M_{a, n} - 1}$
and $\xi_{a, n} := \frac{\alpha_{a, n}}{M_{a, n}}$.
We note that $\xi_{a, n}$ is the mean of the posterior distribution after
observing $n$ entries of arm $a$.
\begin{lemma}\label{lm:dev_xi}
  For all $n\ge 0$,
  \begin{math}
    \abs{\muTape_{a, n} - \xi_{a, n}} \le O\rbr{\frac{M_{a, 0}}{n + M_{a, 0}}}.
  \end{math}
\end{lemma}
\begin{proof}
  After observing $n$ entries, the posterior parameters
  satisfy
  \begin{align*}
    \alpha_{a, n} := \alpha_{a, 0} + \sum_{i\le n} \tape_{a, i},\quad
    \beta_{a, n} := \beta_{a, 0} + \sum_{i\le n} (1 - \tape_{a, i})
    .
  \end{align*}
  It follows that
  \begin{align*}
    \xi_{a, n} =
    \frac{\alpha_{a, 0} + \sum_{i\le n} \tape_{a, i}}{\alpha_{a, 0} + \beta_{a, 0} + n}
    .
  \end{align*}
  Defining $X := \sum_{i\le n} \tape_{a, i}$,
  we can bound the difference between $\xi_{a, n}$ and $\muTape_{a, n}$ as
  \begin{align*}
    \abs{
      \frac{\alpha_{a, 0} + X}{M_{a, 0} +n} - \frac{X}{n}
    }
    &=
    \abs{
      \frac{n\alpha_{a, 0} + nX - nX - X M_{a, 0}}{n(n+M_{a, 0})}
    }
    \\&=
    \abs{
      \frac{n\alpha_{a, 0} - X M_{a, 0}}{n(n+M_{a, 0})}
    }
    \\&\le
    \frac{\alpha_{a, 0}}{n + M_{a, 0}} + \frac{M_{a, 0}}{n+M_{a, 0}}
      &\EqComment{Since $X \le n$}
    \\&\le
    O\rbr{\frac{M_{a, 0}}{n + M_{a, 0}}}
  \end{align*}
\end{proof}
\begin{lemma}\label{lm:dev_rho_xi}
  For all $n\ge 0$,
  \begin{math}
    \abs{\xi_{a, n} - \rho_{a, n}} \le
    O\rbr{\frac{1}{n + M_{a, 0}}}.
  \end{math}
\end{lemma}
\begin{proof}
  \begin{align*}
    \abs{
      \frac{\alpha_{a, n} - 1}{M_{a, n}-1} - \frac{\alpha_{a, n}}{M_{a, n}}
    }
    &=
    \abs{
      \frac{-M_{a, n} + \alpha_{a, n}}{M_{a, n}(M_{a, n}-1)}
    }
    \\&\le
      \frac{M_{a, n}}{M_{a, n}(M_{a, n}-1)}
    \\&= \frac{1}{M_{a, n}-1}
    \\&= O\rbr{\frac{1}{n + M_{a, 0}}}
    &\EqComment{Since $M_{a, n} = M_{a, 0} + n$ and $M_{a, 0} \ge 1$}
  \end{align*}
\end{proof}
We can now prove Theorem \ref{thm:Bayes-unbiased}.
\begin{proof}[Proof of Theorem \ref{thm:Bayes-unbiased}]
  % By Lemma \ref{lm:dev_xi},
  We start with part (a).
  Set $\eta$ to be large enough such that
  % $2\frac{M_{a}}{n}$. It follows that
  \begin{align*}
    \abs{\muTape_{a, n} - \xi_{a, n}} \le \sqrt{\frac{\eta}{n}}.
  \end{align*}
  Since
  \begin{math}
    \frac{M_{a}}{n + M_{a}} \le
    \frac{M_{a}}{n},
  \end{math}
  by Lemma \ref{lm:dev_xi},
  this can be achieved with $\eta \ge O(M_a/\sqrt{N_0})$, which proves part (a).

  For part (b),
  set $\eta$ to be large enough such that
  \begin{math}
    \abs{\muTape_{a, n} - \rho_{a, n}} \le \frac{1}{2} \cdot \sqrt{\frac{\eta}{n}}.
  \end{math}
  Given, Lemmas \ref{lm:dev_xi} and \ref{lm:dev_rho_xi},
  this can be achieved with $\eta \ge  O(M_a/\sqrt{N_0})$.
  Since $M-1 \ge n$,
  we can further gaurantee $\frac{\ln(2/\zeta)}{M-1} \le \frac{\eta}{4n}$
  by setting $\eta \ge O(\ln(1/\zeta))$, which finishes
  the proof together with
  Corollary \ref{cor:concentration_beta}.
\end{proof}